\newcommand\techrep

\ifdefined\techrep
\documentclass[english]{article}
\usepackage{fullpage}
\usepackage[colorlinks,naturalnames,citecolor=blue,urlcolor=blue]{hyperref}
\newcommand{\myhref}[1]{\href{#1}{\url{#1}}}
\else
\documentclass[english,format=acmsmall]{acmart}
\fi

\makeatletter
\DeclareRobustCommand*\cal{\@fontswitch\relax\mathcal}
\makeatother

\usepackage[T1]{fontenc}
\usepackage{color}
\usepackage{amstext}
\usepackage{wasysym}
\newcommand{\perm}{\mathsf{perm}}
\newcommand{\oirel}{\mathsf{oi\_rel}}
\newcommand{\set}{\mathsf{set}}
\newcommand{\loopfree}{\mathsf{loop\_free}}
\usepackage{subfig}
\newcommand{\linter}{\otimes}

\usepackage{algorithm}

\newcommand{\lunion}{\oplus}
\newcommand{\delete}[1]{}

\ifdefined\techrep
\usepackage{amsthm}
  \theoremstyle{plain}
  \newtheorem{theorem}{\protect\theoremname}
  \theoremstyle{definition}
  \newtheorem{definition}{\protect\definitionname}
  \theoremstyle{plain}
  \newtheorem{lemma}{\protect\lemmaname}

  \theoremstyle{definition}
  
  \providecommand{\definitionname}{Definition}
  \providecommand{\examplename}{Example}
  \providecommand{\lemmaname}{Lemma}
  \providecommand{\theoremname}{Theorem}
  \providecommand{\definitionname}{Definition}
  \providecommand{\examplename}{Example}
  \providecommand{\lemmaname}{Lemma}
  \providecommand{\theoremname}{Theorem}
\fi

\makeatletter

\newenvironment{lyxlist}[1]
{\begin{list}{}
{\settowidth{\labelwidth}{#1}
 \setlength{\leftmargin}{\labelwidth}
 \addtolength{\leftmargin}{\labelsep}
 }}
{\end{list}}

\usepackage{tikz}
\usetikzlibrary{arrows,positioning,calc,fit}

\makeatother

\usepackage{babel}

\global\long\def\bool{\mathsf{bool}}

\global\long\def\add{\mathsf{Add}}

\global\long\def\split{\mathsf{Split}}
\global\long\def\id{\mathsf{Id}}

\global\long\def\fb{\mathsf{feedback}}
\global\long\def\In{\mathsf{I}}
\global\long\def\Out{\mathsf{O}}
\global\long\def\Var{\mathsf{V}}

\global\long\def\cvar{\cdot}
\global\long\def\parallel{\mathbin{\|}}
\global\long\def\FB{\mathsf{FB}}
\global\long\def\comp{\;;\,}
\global\long\def\types{\mathsf{Types}}
\global\long\def\dgr{\mathsf{Dgr}}
\global\long\def\dup{\mathsf{Split}}
\global\long\def\sw{\mathsf{Switch}}
\global\long\def\sink{\mathsf{Sink}}
\global\long\def\circlearrow{\stackrel{\circ}{\longrightarrow}}
\global\long\def\partial{\stackrel{c}{\longrightarrow}}
\global\long\def\var{\mathsf{Var}}
\global\long\def\tv{\mathsf{T}}

\global\long\def\arb{\mathsf{Arb}}
\global\long\def\dg{\mathsf{D}}
\global\long\def\Comp{\;;;\,}
\global\long\def\ioeq{\sim}
\global\long\def\Parallel{\,|||\,}
\global\long\def\iodist{\mathsf{io\!-\!distinct}}
\global\long\def\delay{\mathsf{Delay}}
\global\long\def\add{\mathsf{Add}}

\begin{document}

\title{Mechanically Proving Determinacy of Hierarchical Block Diagram Translations}

\ifdefined\techrep
\author{Viorel Preoteasa$^1$ ~~~~ Iulia Dragomir$^2$  ~~~~ Stavros Tripakis$^{1,3}$\\[1ex]
\small$^1$Aalto University, Finland ~~~~  
$^2$Verimag, France ~~~~  
$^3$University of California, Berkeley, USA}
\else
\author{Viorel Preoteasa\inst{1} \and Iulia Dragomir\inst{2} \and Stavros Tripakis\inst{1,3}}
\institute{Aalto University, Finland  \and Verimag, France \and University of California, Berkeley, USA}
\fi

\maketitle

\begin{abstract}
Hierarchical block diagrams (HBDs) are at the heart of embedded system design
tools, including Simulink. Numerous translations exist from HBDs into languages
with formal semantics, amenable to formal verification. However, none of
these translations has been proven correct, to our knowledge.

We present in this paper the first mechanically proven 
HBD translation algorithm.
The algorithm translates HBDs into an algebra of terms
with three basic composition operations (serial, parallel, and feedback).
In order to capture various translation strategies resulting in different
terms achieving different tradeoffs, the algorithm is 
nondeterministic. Despite this, we prove its {\em semantic determinacy}: 
for every input HBD, all possible terms that can be
generated by the algorithm are semantically equivalent.
We apply this result to show how three Simulink translation strategies introduced 
previously can be formalized as determinizations of the algorithm, and derive
that these strategies yield semantically equivalent results
(a question left open in previous work).
All results are formalized and proved in the Isabelle theorem-prover.
\end{abstract}

\section{Introduction}
\label{sec:intro}

Dozens of tools, including Simulink\footnote{\url{http://www.mathworks.com/products/simulink/}}, the most widespread embedded system design environment,
are based on {\em hierarchical block diagrams} (HBDs).
Being a graphical notation (and in the case of Simulink a ``closed'' one
in the sense that the tool is not open-source),
such diagrams need to be translated into other formalisms 
more amenable to formal analysis. Several such translations exist,
e.g., see~\cite{AgrawalSK04,SfyrlaTSBS10, MeenakshiBR06, TripakisSCC05, ReicherdtG2014, ChenDS09, YangV12, ZhouK12, ZouZWFQ13, MinopoliF16} and related
discussion in~\S\ref{sec:rwork}.
To our knowledge, none of these translations has been formally verified.
This paper aims to remedy this fact.

This work is part of a larger project, the Refinement Calculus of Reactive
Systems (RCRS) --
see~\cite{RCRS_Toolset_TACAS_2018,RCRS_PDTarxiv2017,preoteasa:tripakis:2014,TripakisTOPLAS2011} and
\myhref{http://rcrs.cs.aalto.fi/}.
RCRS is a compositional framework for modeling and reasoning about reactive
systems. It allows to capture systems which can be both non-deterministic
and non-input-receptive, and offers compositional refinement and other
features for modular specification and verification.
 RCRS comes with a toolset~\cite{RCRS_Toolset_TACAS_2018} which includes a full implementation
of the RCRS theory and related analysis tools on top of the Isabelle theorem prover~\cite{nipkow-paulson-wenzel-02}, and
a Translator of Simulink diagrams to RCRS theories.

The Translator, first described in~\cite{DBLP:conf/spin/DragomirPT16},
implements three translation strategies from HBDs to an algebra
of components with three basic composition operators: serial, parallel,
and feedback. The several translation strategies are motivated by the fact that
each strategy has its own pros and cons. For instance, one strategy
may result in shorter and/or easier to understand algebra terms, while
another strategy may result in terms that are easier to simplify by manipulating
formulas in a theorem prover.
But a fundamental question is left open in~\cite{DBLP:conf/spin/DragomirPT16}: are these translation strategies {\em semantically equivalent}, meaning,
do they produce semantically equivalent terms? This is the question we study and answer (positively) in this paper.

The question is non-trivial, as we seek to prove the equivalence of three
complex algorithms which manipulate a graphical notation (hierarchical block
diagrams) and transform models in this notation into a different textual
language, namely, the algebra mentioned above. Terms in this algebra have
intricate formal semantics, and formally proving that two given specific
terms are equivalent is already a non-trivial exercise. Here, the problem
is to prove that a number of translation strategies $T_1,T_2,...,T_k$
are equivalent, meaning that for {\em any} given graphical diagram $D$, 
the terms resulting from translating $D$ by applying these strategies,
$T_1(D), T_2(D), ..., T_k(D)$, are all semantically equivalent.

This equivalence question is important for many reasons.
Just like a compiler has many choices when generating code, a HBD translator
has many choices when generating algebraic expressions. Just like a correct
compiler must guarantee that all possible results are equivalent
(independently of optimization or other flags/options), the translator must
also guarantee that all possible algebraic expressions are equivalent. 
Moreover, the algebraic expressions constitute the formal semantics of HBDs,
and hence also those of tools like Simulink. Therefore, this determinacy
principle is also necessary in order for the formal Simulink semantics to be well-defined.

In order to formulate the equivalence question precisely,
we introduce an {\em abstract} and {\em nondeterministic} algorithm for translating 
HBDs into an abstract algebra of components with three composition operations
(serial, parallel, feedback) and three constants (split, switch, and sink).
By {\em abstract algorithm} we understand an algorithm that produces terms in this abstract algebra.
Concrete versions for this algorithm are
obtained when using it for concrete models of the algebra
(e.g., {\em constructive functions}).
The algorithm is {\em nondeterministic} in the sense that it consists of a set
of basic operations (transformations) that can be applied in any order.
This allows to capture various deterministic translation strategies as
determinizations ({\em refinements}~\cite{back-wright-98}) of the abstract algorithm.

The main contributions of the paper are the following:
\begin{enumerate}
\item We formally and mechanically define a translation algorithm for HBDs.
\item We prove that despite its internal nondeterminism, the 
algorithm achieves deterministic results in the sense that all possible algebra
terms that can be generated by the different nondeterministic choices are
semantically equivalent.
\item We formalize two translation strategies introduced in~\cite{DBLP:conf/spin/DragomirPT16} as refinements of the abstract algorithm.
\item We formalize also the third strategy (feedbackless)
introduced in~\cite{DBLP:conf/spin/DragomirPT16} as an independent algorithm. 
\item We prove the equivalence of these three translation strategies.
\end{enumerate}
To our knowledge, our work constitutes the first and only mechanically proven hierarchical block diagram translator.

We remark that our translation is compositional~\cite{DBLP:conf/spin/DragomirPT16}.
We also remark that our abstract algorithm can be instantiated in
many different ways, encompassing not just the three translation strategies
of~\cite{DBLP:conf/spin/DragomirPT16}, but also any other HBD translation
strategy that can be devised by combining the basic composition operations 
defined in the abstract algorithm. As a consequence, our results imply not
just the equivalence of the translation strategies of~\cite{DBLP:conf/spin/DragomirPT16}, but also the equivalence of any other translation strategy that 
could be devised as stated above. More generally, any translation
of a graphical formalism into expressions in some language would
have to deal with problems similar to those tackled in this paper, and
our work offers an example of how to address these problems in a formal manner.

The entire RCRS framework, including all results in this paper, have been formalized and proved in the Isabelle theorem prover~\cite{nipkow-paulson-wenzel-02} and are part of the RCRS toolset which is publicly available in a figshare repository~\cite{RCRS_Toolset_figshare}.
The theories relevant to this paper are under RCRS/Isabelle/TranslateHBD.
The RCRS toolset can be downloaded also from the RCRS web page: \url{http://rcrs.cs.aalto.fi/}.

\section{Related Work}
\label{sec:rwork}

Model transformation and the verification of its correctness is a long standing line of research, which includes classification of model transformations~\cite{JOT:AmraniCLSDTVC15} and the properties they must satisfy with respect to their intent~\cite{LucioADLSSSW16}, verification techniques~\cite{AbRahimL15}, frameworks for specifying model transformations (e.g., ATL~\cite{ATL}), and various implementations for specific source and target meta-models. Extensive surveys of the above can be found in~\cite{JOT:AmraniCLSDTVC15,CALEGARI20135,AbRahimL15}.

Several translations from Simulink have been proposed 
in the literature, including to Hybrid Automata~\cite{AgrawalSK04}, BIP~\cite{SfyrlaTSBS10}, NuSMV~\cite{MeenakshiBR06},  Lustre~\cite{TripakisSCC05}, Boogie~\cite{ReicherdtG2014}, Timed Interval Calculus~\cite{ChenDS09}, Function Blocks~\cite{YangV12}, I/O Extended Finite Automata~\cite{ZhouK12}, Hybrid CSP~\cite{ZouZWFQ13}, and SpaceEx~\cite{MinopoliF16}. 
It is unclear to what extent these approaches provide formal guarantees on the
determinism of the translation. For example, the order in which blocks in the
Simulink diagram are processed might a-priori influence the result.
Some works fix this order, e.g., \cite{ReicherdtG2014} computes the control flow graph 
and translates the model according to this computed order. In contrast, 
we prove that the results of our algorithm are equivalent for any order.
To the best of our knowledge, the abstract translation proposed hereafter for Simulink is 
the only one formally defined and mechanically proven correct.

The focus of several works is to validate the preservation of the semantics
of the original diagram by the resulting translation (e.g., see~\cite{YangV12,RyabtsevS2009,BouissouC12,Schlesinger2016}).
In contrast, our goal is to prove equivalence of all possible translations.
Given that Simulink semantics is informal (``what the simulator does''),
ultimately the only way to gain confidence that the translation conforms
to the original Simulink model is by simulation (e.g., as in~\cite{DBLP:conf/spin/DragomirPT16}).

The work of \cite{abramski:1994} presents a correspondence between formulas and proofs in linear logic
\cite{girard:1987}
and types and computations in process calculi \cite{Hoare:1978:CSP,Milner:1989:CC}. A sequent in the logic ($\vdash A,C^{\bot},B$) is interpreted as an interface specification for a concurrent process and how this process
is connected to the environment. In this example $A$ and $B$ are inputs to the process and $C^{\bot}$
is output. In our approach we connect components by naming their inputs and outputs, and an output
is connected to an input if they have the same name.

With respect to the target algebra of our translation, the most relevant related works are 
the algebra of flownomials \cite{Stefanescu:2000:NA:518304} and the relational model for 
non-deterministic dataflow \cite{hildebrandt2004}. A comparison with these works is presented 
in Section~\ref{sec:model}.

In \cite{Courcelle1987}, graphs and graph operations are represented by algebraic expressions 
and operations, and a complete equational axiomatization of the equivalence of the graph expressions 
is given. This is then applied to flow-charts as investigated in \cite{SCHMECK1983165}.

To our knowledge, none of the theoretical works on flownomials, nor graph represented as expressions, nor the more practical works on translating HBDs/Simulink, are mechanically formalized or verified.

\section{Preliminaries}
\label{sec:prelim}

For a type or set $X$, $X^{*}$ is the type of finite lists with elements from
$X$. We denote the empty list by $\epsilon$, $(x_{1},\ldots,x_{n})$
denotes the list with elements $x_{1},\ldots,x_{n}$, and for lists
$x$ and $y$, $x\cvar y$ denotes their concatenation. 
The length
of a list $x$ is denoted by $|x|$. The list of common elements of 
$x$ and $y$ in the order occurring in $x$ is denoted by $x\linter y$.
The list of elements from $x$ that do not occur in $y$ is denoted by 
$x\ominus y$. We define $x\lunion y = x\cdot(y\ominus x)$, the list 
of $x$ concatenated with the elements of $y$ not occurring in $x$.
A list $x$ is a {\em permutation} of a list $y$, denoted $\perm(x,y)$, if $x$ contains  
all elements of $y$ (including multiplicities) possibly in a different 
order. For a list $x$, $\set(x)$ denotes the set of all elements of $x$.

\subsection{Constructive Functions}
\label{subsec:constr_func}

We introduce in this section the {\em constructive functions} as used in the {\em constructive semantics} literature~\cite{Malik94,Berry99,EdwardsLee03}.
They will provide a concrete model for the abstract algebra of HBDs, introduced in Section \ref{sec:model}. These functions are also used in the example from Section \ref{sec:overview}.

We assume that $\bot$ is a new element called unknown, and that 
$\bot$ is not an element of other sets that we use. For a set $A$ we
define $A^\bot = A \cup \{\bot\}$, and on $A^\bot$ we introduce
the {\em pointed complete partial order} (cpo) \cite{priestley} by 
$(a\le b) \iff (a=\bot \lor a = b)$. We extend the order on $A^\bot$
to the Cartesian product $A^\bot_1\times\cdots A^\bot_n$ by 
$(x_1,\ldots, x_n) \le (y_1,\ldots,y_n) \iff (\forall 1\le i\le n : x_i \le y_i)$.

Constructive functions are the {\em monotonic} functions $f: A_{1}^{\bot}\times\ldots\times A_{n}^{\bot}\to  B_{1}^{\bot}\times\ldots\times B_{m}^{\bot}$, i.e., ($\forall x,y: x\le y \Rightarrow f(x) \le f(y)$).
We denote these functions by $A_{1}\cdots A_{n}\partial B_{1}\cdots B_{m}$ ($f:A_{1}\cdot\ldots\cdot A_{n}\partial B_{1}\cdot\ldots\cdot B_{m}$ for $f: A_{1}^{\bot} \times \ldots \times A_{n}^{\bot}\to  B_{1}^{\bot}\times\ldots\times B_{m}^{\bot}$). $\id :A \partial A$ denotes the {\em identity
function} on $A$: $\forall x: \id(x) = x$.

For constructive functions $f:A\partial B$ and $g:B\partial C$, their
{\em serial composition} $g\circ f$ is the normal function composition $(g\circ f)(x) = g(f(x))$. The {\em parallel composition} of
$f:A\partial B$ and $g:A'\partial B'$ is denoted 
$f\parallel g:A\cdot A'\partial B\cdot B'$ and is defined by $(f\parallel g)(x,y) = (f(x),g(y))$. We assume that parallel composition operator binds stronger than serial composition, i.e. $f \parallel g\circ h$ is the same as $(f \parallel g)\circ h$.

For a constructive function $f:A\partial A$, its least fixpoint always exists \cite{priestley}, 
and we use it to define a {\em feedback composition}.
If $f:A\cdot B\partial A\cdot B'$
is a constructive function, then its feedback (on $A$), denoted $\fb(f):B\partial B'$, is defined by
\[
\fb(f)(y)=f(\mu\ x: f_{1}(x, y), y)
\]
where $f_{1}:A\cdot B\partial A$ is the first component of $f$
and $(\mu\ x:  f_{1}(x, y))$ is the least fixpoint of the function that,
for fixed $y$, maps $x$ into $f_{1}(x, y)$.

Let $x_1,\ldots,x_n$ be variables ranging over types $A_1,\ldots, A_n$,
and $e_1,\ldots,e_m$ expressions using basic operations ($+$,$-$,$\ldots$) on these variables, ranging over
types  $B_1,\ldots, B_m$. We define the constructive function
$$
[x_1,\ldots, x_n \leadsto e_1,\ldots,e_m]: A_1\cdots A_n
\partial B_1\cdots B_n
$$
as the function that maps $(x_1,\ldots,x_n) \in A_1^\bot\times\ldots 
\times A_n^\bot$ into $(e_1,\ldots,e_m)$, where the basic operations
are extended to unknown values in a standard way (e.g.\ $3+\bot = \bot$, 
$\bot\cdot 0 = 0$).

\subsection{Refinement Calculus and Hoare Total Correctness Triples}

We model the (nondeterministic) algorithms using monotonic
predicate transformers \cite{dijkstra:1975} within the refinement calculus 
\cite{back-wright-98}.

We assume a state space $\Sigma$. A state $\sigma\in\Sigma$ gives values to all
program variables. Programs are modeled as monotonic predicate transformers on
$\Sigma$, that is monotonic functions from predicates to predicates 
($(\Sigma \to \bool)\to(\Sigma \to \bool)$) with a weakest precondition interpretation.
For $P: (\Sigma \to \bool)\to(\Sigma \to \bool)$ and a post condition $q:\Sigma \to \bool$,
$P(q)$ is the predicate that is true for the initial states from which the execution
of the program modeled by $P$ always terminates, and it terminates in a state from $q$.
In the rest of the paper we refer to monotonic elements of 
$(\Sigma \to \bool)\to(\Sigma \to \bool)$
as programs. The program statements are modeled as operations on monotonic predicate 
transformers.

For predicates ($\Sigma\to\bool$), we use the notations $\cup$, $\cap$, $\neg$,
and $\subseteq$ for the union, intersection, complement, and subset operations, respectively.

The {\em nondeterministic assignment statement}, denoted $[\,x:=x'\ |\ p(y,x')\,]$, assigns
a new value $x'$ to variable $x$ such that the property $p(y,x')$ is true. 
In $p(y,x')$, variable $y$ stands for the current value (before the assignement) of $y$ used 
for updating varible $x$. We can choose $y = x$, to refer to the current value of variable $x$.
For example
$[\,x:=x'\ | \ x' > x + 1\,]$ assigns to $x$ a new value greater that the current value of $x+1$.

Formally, the nondeterministic assignment statement is defined by:
$$
[\,x:=x'\ |\ p(y,x')\,](q)(\sigma) = (\forall x': p(\sigma(y), x') \Rightarrow q(\sigma[x:=x']))
$$
where $\sigma(y)$ is the value of variable $y$ in state $\sigma$, and $\sigma[x:=x']$ is a new
state obtained from $\sigma$ by changing the value of $x$ to $x'$.

The {\em standard assignment statement} $x:=e$ is defined as $[\,x:=x'\ | \ x' = e]$, 
where $x'$ is a new name.

For a predicate $p:\Sigma\to\bool$, the {\em assert statement}, denoted $\{p\}$, 
starting from a state $\sigma$ behaves as skip if $p(\sigma)$ is true, and it {\em fails} otherwise.
By fail we mean a program that runs forever. 
$$
\{p\}(q) = p \cap q
$$

The {\em sequential composition} of programs $P, P'$, denoted $P\comp P'$ is the function composition 
of predicate transformers:
$$(P \comp P') (q) = P(P'(q)).$$

The {\em nondeterministic choice} of $P$ and
$P'$, denoted $P\sqcap P'$, is the pointwise extension of the intersection on predicates to 
predicate transformers: $$(P\sqcap P')(q) = P(q) \cap P'(Q).$$

For a predicate $b$ and programs $P$ and $P'$, the {\em if statement}, denoted 
$\mathsf{if\ } b \mathsf{\ then\ } P \mathsf{\ else\ } P'$ is defined by
$$
\mathsf{if\ } b \mathsf{\ then\ } P \mathsf{\ else\ } P' = (\{b\} \comp P)\sqcup(\{\neg b\} \comp P')
$$
where $\sqcup$ is the pointwise extension of union on predicates to predicate transformers ($(P\sqcup P')(q) = 
P(q)\cup P'(q)$). 

For predicate $b$ and program $P$, the {\em while statement}, denoted $\mathsf{while\ }b
\mathsf{\ do\ } P$, is defined as a least fixpoint:
$$
\mathsf{while\ }b
\mathsf{\ do\ } P = (\mu\, X: \mathsf{if\ } b \mathsf{\ then\ } P\comp X \mathsf{\ else\ skip})
$$
where $\mathsf{skip}$ is the program that does not change the state, modeled as the
identity predicate transformer, and $(\mu\, X: \mathsf{if\ } b \mathsf{\ then\ } P\comp X \mathsf{\ else\ skip})$ is the least fixpoint of the function mapping $X$ into 
$\mathsf{if\ } b \mathsf{\ then\ } P\comp X \mathsf{\ else\ skip}$. The fixpoint always exists
because of the monotonicity assumption.

The {\em refinement relation of programs}, denoted $P\sqsubseteq P'$, is again the pointwise extension
of the inclusion order on predicates to predicate transformers:
$$
(P\sqsubseteq P') = (\forall q: P(q) \subseteq P'(q)).
$$
If a program $P'$ is a refinement of another program $P$, ($P\sqsubseteq P'$), then we can use $P'$
in every context where we can use $P$. In a refinement $P\sqsubseteq P'$, the program $P'$ is more
deterministic than $P$, and it fails for less input states. For example we have the following refinement:
$$
\{x > 10\}\comp[ x:= x' \ | \ x' = 1 \lor x' = 2 \lor x' = 3] \sqsubseteq 
\{x > 0\}\comp[ x:= x' \ | \ x' = 1 \lor x' = 3]
$$
In this example, the second program fails for less states $x > 0$ as opposed to $x > 10$,
and it is more deterministic. The second program can assign to $x$ only the values $1$ and $3$,
while the first program can assign also value $2$.

Finally we introduce Hoare \cite{hoare-69} {\em total correctness triples} for programs. If $p$
is a {\em precondition} predicate on states, $q$ is a {\em postcondition} predicate on states, and
$P$ is a program, then the Hoare total correctness triple $p\ \{\!|P|\!\}\ q$ is defined by
$$
(p\ \{\!|P|\!\}\ q) = (p\subseteq P(q)).
$$
The interpretation of the triple $p\ \{\!|P|\!\}\ q$ is the following. If the program $P$ starts
from a state $\sigma$ satisfying the precondition $p$, then $P$ always terminates, and it terminates
into a state satisfying the postcondition $q$. 

In general, the correctness of a program is stated as a Hoare triple, and it is proved by reducing 
this correctness problem to smaller and smaller programs using Hoare rules.
As examples we
give here two Hoare rules for the correctness of the nondeterministic assignment and while 
statements.\footnote{We omit several of the proofs of the results presented in this paper. These proofs and additional material used in the formalization and verification of our algorithms can be found in our Isabelle formalization~\cite{RCRS_Toolset_figshare}.}

\begin{lemma}[Hoare rule for the nondeterministic assignment]
\label{hoare:assign}
If $p,q$ are predicates on state
and $b$ is a predicate on $y,x'$ such that
$$
(\forall \sigma, x' : p(\sigma) \land b(\sigma(y),x') \Rightarrow q(\sigma[x:=x']))
$$
then
$$
(p\ \{\!|\ [\, x:=x'\ |\ b(y,x')\,]\ |\!\}\ q).
$$
\end{lemma}

\newcommand{\nat}{\mathsf{nat}}
\begin{lemma}[Hoare rule for while]
\label{hoare:while}
If $p,q,b,I$ are predicates
on state, $t:\Sigma\to\nat$ is a function from state to natural numbers, 
and $P$ is a program such that
$$
(\forall n: (I \land t = n) \ \{\!|P|\!\}\ (I\land t < n) ) \mbox{\ \ and \ \ } (p\subseteq I)
\mbox{\ \ and\ \ } (\neg b \cap I \subseteq q)
$$
then
$$
(p\ \{\!|\ \mathsf{while\ }b\mathsf{\ do\ } P\ |\!\}\ q)
$$
\end{lemma}
In this lemma $I$ is called the {\em invariant} and its role is to ensure the correctness of the
while program based on the correctness of the body ($P$). The function (term) $t$ is called 
the {\em variant}
and it is used to ensure the termination of the while program for all possible input states 
satisfying $p$.

There is an important relationship between Hoare rules and refinement, expressed by the next lemma.
\begin{lemma}\label{hoare:refinement} If $P, P'$ are programs, then
$$
P\sqsubseteq P' \ \ 
\Leftrightarrow\ \ 
(\forall p,q: (p\ \{\!|\ P\ |\!\}\ q) \ \Rightarrow (p\ \{\!|\ P'\ |\!\}\ q))
$$
\end{lemma}

\mbox{}

\section{Overview of the Translation Algorithm}
\label{sec:overview}

A \emph{block diagram} $N$ is a network of interconnected blocks. A block
may be a basic ({\em atomic}) block, or a {\em composite} block
that corresponds to a {\em sub-diagram}. If $N$ contains composite blocks
then it is called a \emph{hierarchical block diagram} (HBD); otherwise
it is called {\em flat}.
An example of a flat diagram is shown in Figure \ref{fig:sum-orig}.
The connections between blocks are called wires, and they have a source
block and a target block. For simplicity, we will assume that every wire
has a single source and a single target. This can be achieved by adding
extra blocks. For instance, the diagram of Figure \ref{fig:sum-orig} can
be transformed as in Figure \ref{fig:sum-name} by adding an explicit
block called {\em Split}.

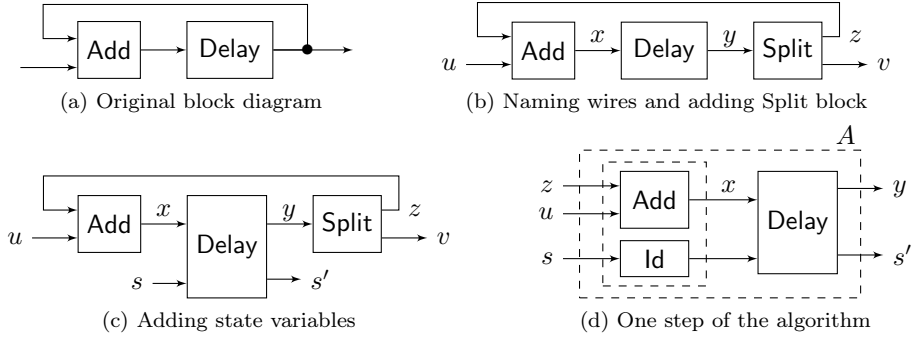
\begin{figure}
\centering
\subfloat[Original block diagram]
{
	\begin{tikzpicture} 		
    \node[draw, minimum height = 5ex](Add){$\mathsf{Add}$};
    \node[draw, minimum height = 5ex,right = 4ex of Add,text width=6ex,text centered](Delay){$\mathsf{Delay}$
    };
	\draw[-latex'](Delay.east)-- ++(3ex,0) node(C){$\bullet$} -- ++(0ex,4ex) coordinate(A) 
		-- (A -| Add.west) -- ++ (-3ex, 0) -- ++(0,-3ex) coordinate(B) -- (B -| Add.west);
    \draw[-latex'](Add)--(Delay);
	\draw[-latex'](Add.south west) ++ (-5ex,1ex) coordinate (A) -- (A-| Add.west);
	\draw[-latex'](Delay.east) -- ++(7ex, 0);
	\end{tikzpicture}
	\label{fig:sum-orig}
	}\qquad
\subfloat[Naming wires and adding Split block]
{
	\begin{tikzpicture} 		
    \node[draw, minimum height = 5ex](Add){$\mathsf{Add}$};
	\node[right = 0.5ex of Add.north east, anchor = north west]{$x$};
    \node[draw, minimum height = 5ex,right = 4ex of Add,text width=6ex,text centered](Delay){$\mathsf{Delay}$
    };
	\node[right = 0.5ex of Delay.north east, anchor = north west]{$y$};
    \node[draw, minimum height = 5ex,right = 4ex of Delay](Split){$\mathsf{Split}$};
	\node[right = 1.5ex of Split.north east, anchor = north west](z){$z$};

    \node[left = 4ex of Add.south west, anchor = south east](u){$u$};
    \node[right = 4ex of Split.south east, anchor = south west](v){$v$};

	\draw[-latex'](z -| Split.east) -- (z.west) -- ++(0ex,3ex) coordinate(A)
		-- (A -| Add.west) -- ++ (-3ex, 0) -- ++(0,-3ex) coordinate(B) -- (B -| Add.west);
    \draw[-latex'](Add)--(Delay);
	\draw[-latex'](Delay.east) -- (Split.west);
	\draw[-latex'](v -| Split.east) -- (v);
	\draw[-latex'](u) -- (u -| Add.west);
	\end{tikzpicture}
	\label{fig:sum-name}
	} \\
\subfloat[Adding state variables]
{
	\begin{tikzpicture} 		
    \node[draw, minimum height = 5ex](Add){$\mathsf{Add}$};
	\node[right = 0.5ex of Add.north east, anchor = north west]{$x$};
    \node[draw, minimum height = 9ex,right = 4ex of Add.north east, anchor = north west](Delay){$\mathsf{Delay}$};
	\node[right = 0.5ex of Delay.north east, anchor = north west]{$y$};
    \node[draw, minimum height = 5ex,right = 4ex of Delay.north east, anchor = north west](Split){$\mathsf{Split}$};
	\node[right = 1.5ex of Split.north east, anchor = north west](z){$z$};

    \node[left = 4ex of Add.south west, anchor = south east](u){$u$};
    \node[right = 4ex of Split.south east, anchor = south west](v){$v$};

    \node[left = 3ex of Delay.south west, anchor = south east](s){$s$};
    \node[right = 3ex of Delay.south east, anchor = south west](s1){$s'$};

	\draw[-latex'](z -| Split.east) -- (z.west) -- ++(0ex,3ex) coordinate(A)
		-- (A -| Add.west) -- ++ (-3ex, 0) -- ++(0,-3ex) coordinate(B) -- (B -| Add.west);
    \draw[-latex'](Add)--(Add -| Delay.west);
	\draw[-latex'](Split.west -| Delay.east) -- (Split.west);
	\draw[-latex'](v -| Split.east) -- (v);
	\draw[-latex'](u) -- (u -| Add.west);
	\draw[-latex'](s1 -| Delay.east) -- (s1);
	\draw[-latex'](s) -- (s -| Delay.west);
	\end{tikzpicture}
	\label{fig:sum-state}
	} \qquad 
\subfloat[One step of the algorithm]
{
	\begin{tikzpicture} 		
    \node[draw, minimum height = 5ex, minimum width = 6ex](Add){$\mathsf{Add}$};
	\node[right = 2ex of Add.north east, anchor = north west](x){$x$};

    \node[draw, minimum width = 6ex, below = 1ex of Add](Id){$\mathsf{Id}$};

    \node[left = 5ex of Add.south west, anchor = south east](u){$u$};
    \node[left = 5ex of Add.north west, anchor = north east](z){$z$};

    \node[draw, minimum height = 9ex,right = 6ex of Add.north east, anchor = north west](Delay){$\mathsf{Delay}$};
	\node[right = 4ex of Delay.north east, anchor = north west](y){$y$};
    \node[right = 4ex of Delay.south east, anchor = south west](s1){$s'$};

    \node[left = 5ex of Id.west](s){$s$};
    \node[draw, dashed, fit = (Id) (Add), inner xsep = 1.5ex](Par){};

    \draw[-latex'](s) -- (Id);
    \draw[-latex'](z) -- (z -| Add.west);
    \draw[-latex'](u) -- (u -| Add.west);
	\draw[-latex'](Add) -- (Add -| Delay.west);

    \draw[-latex'](y -| Delay.east) -- (y);
    \draw[-latex'](s1 -| Delay.east) -- (s1);
    \draw[-latex'](Id) -- (Id -| Delay.west);
    
    \node[above right = 1.5ex and -1ex of Delay](A){$A$};
    \node[draw, dashed, fit = (Par) (Delay), inner ysep = 1ex, inner xsep = 2ex](){};
	\end{tikzpicture}
	\label{fig:one-step}
	}
\caption{Running example: diagram for summation.}
\label{fig:sum}
\end{figure}

Let us explain the idea of the translation algorithm.
We first explain the idea for flat diagrams, and then
we extend it recursively for hierarchical diagrams. 

A diagram is represented in the algorithm as a list of elements corresponding
to the basic blocks. One element of this list is a triple containing
a list of input variables, a list of output variables, and a \emph{function}.
The function computes the values of the outputs based
on the values of the inputs, and for now it can be thought of as a
constructive function. Later this function will be an element of an
abstract algebra modeling HBDs. Wires are represented by matching
input/output variables from the block representations.

A block diagram may contain {\em stateful} blocks such as delays or integrators.
We model these blocks using additional state variables (wires). 
In Figure \ref{fig:sum}, the only stateful block is the block $\mathsf{Delay}$.
\delete{
\red{where the current and next states are denoted by $s$ 
and $s'$ respectively}. \red{Intuitively, a $\mathsf{Delay}$ models a unit delay of the output, that is at any moment $n$ the output is equal 
to the input at moment $n-1$. Then the values to remember are stored in the current and next state variables.}
}
We model this block as an element with two inputs $(x,s)$, two outputs $(y,s')$ and function $(y,s'):=(s,x)$ (Figure \ref{fig:sum-state}). More details
about this representation can be found in \cite{DBLP:conf/spin/DragomirPT16}.

In summary, the list representation of the example of Figure \ref{fig:sum}
is the following:
$$
\begin{array}{lll}
\big(\mathsf{Add},\mathsf{Delay},\mathsf{Split}\big)\ \mbox{where} \\[1ex]
\add  =  ((z,u),\ x,\ [z,u\leadsto z+u])\\[1ex]
\delay  =  ((x,s),\ (y,s'),\ [x,s\leadsto s,x])\\[1ex]
\split  =  (y,\ (z,v),\ [y\leadsto y,y])
\end{array}
$$

The algorithm works by choosing nondeterministically some elements
from the list and replacing them with their appropriate composition (serial,
parallel, or feedback). The composition must
connect all the matching variables. Let us illustrate how the algorithm may
proceed on the example of Figure \ref{fig:sum}; for the full description of
the algorithm see Section \ref{sec:algorithm}. 

Suppose the algorithm first chooses to compose $\add$ and $\delay$.
The only matching variable in this case is $x$, between the output
of $\add$ and the first input of $\delay$. The appropriate composition
to use here is serial composition. Because $\delay$ also has
$s$ as input, $\add$ and $\delay$ cannot be directly connected in series.
This is due to the number of outputs of $\add$ that need to match the number of inputs
of $\delay$. To compute the serial composition, $\add$ must first be composed in parallel with
the identity block $\id$, as shown in Figure~\ref{fig:one-step}. Doing so,
a new element $A$ is created: 
\delete{$$\red{A=((z,u,s),\ (y,s'),\ [x,s\leadsto s,x]\circ([u,x\leadsto u+z\parallel\id)).}$$}
$$A=((z,u,s),\ (y,s'),\ \delay \circ (\add\parallel\id)).$$
Next, $A$ is composed with $\split$. In this case we need to connect variable $y$
(using serial composition), as well as $z$ (using feedback composition).
The resulting element is
\[
\Big((u,s),\ (v,s'),\ \fb\big((\split\parallel\id)\circ\delay\circ(\add\parallel\id)\big)\Big)
\]
where we need again to add the $\id$ component for variable $s'$.

Suppose now that the algorithm chooses to compose first the blocks $\split$ and $\add$
(Fig.~\ref{fig:splitadd}) into $B$.
$$
B = ((y,u),\ (x,v), \ (\add\ \|\ \id) \circ (\id\ \|\ [v,u \leadsto u,v]) \circ (\split\ \|\ \id))
$$
In this composition, in addition to the $\id$ components, we need now also a switch 
($[v,u \leadsto u,v]$) for wires $v$ and $u$.
Next the algorithm composes $B$ and $\delay$ (Fig.~\ref{fig:Bdelay}):
$$
\Big((u,s),\ (s',v),\ \fb\big((\delay \ \| \ \id) \circ (\id\ \|\ [v,s \leadsto s,v]) \circ (B \ \| \ \id)\big)
\Big)
$$

As we can see from this example, by considering the blocks in the diagram in different orders, we obtain 
different expressions. On this example, the first expression is simpler (it has less connectors) than the second one. 
In general, a diagram, being a graph, does not have a predefined canonical order, and we need to show that the result of the algorithm is {\em the same} regardless of the order in which the blocks are considered. 
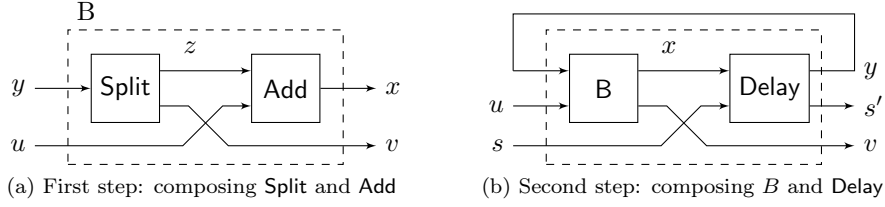
\begin{figure}
\centering
\subfloat[First step: composing $\split$ and $\add$ \label{fig:splitadd}]
{
\begin{tikzpicture} 		
    \node[draw, minimum height = 6ex, minimum width = 6ex](Split){$\mathsf{Split}$};
    \node[draw, minimum height = 6ex, minimum width = 6ex, right = 8ex of Split](Add){$\mathsf{Add}$};
    \node[below = 1ex of Split](t){};
     \node[right = 2ex of Split.north east, anchor = south west, inner sep = 0](z){$z$};
    \draw[-latex'](Split.north east) ++ (0, -1.5ex) coordinate(A) -- (A -| Add.west);
    \draw[-latex'](Split.south east) ++ (0, 1.5ex) coordinate(B) -- 
    	++(2.5ex, 0) -- ++(3.5ex, -3.5ex) -- ++(13ex, 0) coordinate(B);
    \draw[latex'-](Add.south west) ++ (0, 1.5ex) -- 
    	++(-2.5ex, 0) -- ++(-3.5ex, -3.5ex) -- ++(-13ex, 0) coordinate(C);
    \draw[-latex'](Add.east) -- (Add.east -| B) coordinate(D);
    \draw[latex'-](Split.west) -- (Add.west -| C) coordinate(E);
    \node[right = 0ex of D]{$x$};
    \node[right = 0ex of B]{$v$};
    \node[left = 0ex of E]{$y$};
    \node[left = 0ex of C]{$u$};
    \node[draw, dashed, fit = (Split) (Add) (z) (t), inner ysep = 1ex, inner xsep = 2ex](BB){};
    \node[above = 0ex of BB.north west, anchor = south west]{B};
\end{tikzpicture}
}
\qquad
\subfloat[Second step: composing $B$ and $\delay$ \label{fig:Bdelay}]
{
\begin{tikzpicture} 		
    \node[draw, minimum height = 6ex, minimum width = 6ex](BB){$\mathsf{B}$};
    \node[draw, minimum height = 6ex, minimum width = 6ex, right = 8ex of BB](Delay){$\mathsf{Delay}$};
    \node[below = 1ex of BB](t){};
     \node[right = 2ex of BB.north east, anchor = south west, inner sep = 0](x){$x$};
     
    \draw[-latex'](BB.north east) ++ (0, -1.5ex) coordinate(A) -- (A -| Delay.west);
    
    \draw[-latex'](BB.south east) ++ (0, 1.5ex) -- 
    	++(2.5ex, 0) -- ++(3.5ex, -3.5ex) -- ++(13ex, 0) coordinate(B);
    \draw[latex'-](Delay.south west) ++ (0, 1.5ex) -- 
    	++(-2.5ex, 0) -- ++(-3.5ex, -3.5ex) -- ++(-13ex, 0) coordinate(C);
    	
    \draw[-latex'](Delay.south east)++(0, 1.5ex)coordinate(X) -- (X -| B) coordinate(F);
    \draw[latex'-](BB.north west)++(0, -1.5ex) coordinate(X)  -- (X -| C) coordinate(E);
    \draw[latex'-](BB.south west)++(0, 1.5ex) coordinate(X)  -- (X -| C) coordinate(G);
    \node[right = 0ex of F]{$s'$};
    \node[right = 0ex of B]{$v$};
    \node[left = 0ex of G]{$u$};
    \node[left = 0ex of C]{$s$};
    \draw[-](Delay.north east)++(0, -1.5ex)coordinate(X) -- (X -| B) coordinate(D)
    	-- ++(0, 5ex) coordinate(X) -- (X -| C) -- (E);
    \node[right = 0ex of D]{$y$};
    \node[draw, dashed, fit = (Split) (Add) (x) (t), inner ysep = 1ex, inner xsep = 2ex](BB){};
\end{tikzpicture}
}
\caption{A different composition order for the example from Fig.~\ref{fig:sum}.}
\end{figure}

We make two remarks here. First, the final result of the algorithm is a triple with the 
same structure
as all elements on the original list: (input variables, output variables, function), where the function
represents the computation performed by the entire diagram. Therefore,
the algorithm can be applied recursively on HBDs.
	
Second, the variables in the representation
occur at most twice, once as input, and once as output. The variables
occurring only as inputs are the inputs of the resulting final element,
and variables occurring only as outputs are the outputs of the resulting
final element. This is true in general for all diagrams, due to the
representation of splitting of wires. This fact is essential for the
correctness of the algorithm as we will see in Section \ref{sec:algorithm}.

\section{An Abstract Algebra for Hierarchical Block Diagrams}
\label{sec:model}

We assume that we have a set of $\types$.
We also assume a set of {\em diagrams} $\dgr$. Every element $S\in\dgr$
has input type $t\in\types^{*}$ and output type $t'\in\types^{*}$.
If $t=t_{1}\cdots t_{n}$ and $t'=t'_{1}\cdots t'_{m}$,
then $S$ takes as input a tuple of the type $t_{1}\times\ldots\times t_{n}$
and produces as output a tuple of the type $t'_{1}\times\ldots\times t'_{m}$.
We denote the fact that $S$ has input type $t\in\types^{*}$ and
output type $t'\in\types^{*}$ by $S:t\circlearrow t'$. The elements
of $\dgr$ are abstract. 

\subsection{Operations of the Algebra of HBDs}

\paragraph{Constants.}
Basic blocks are modeled as constants on $\dgr$.
For types $t,t'\in\types^{*}$ we assume the following constants:
\[
\begin{array}{l}
\id(t):t\circlearrow t\\[1ex]
\dup(t):t\circlearrow t\cdot t\\[1ex]
\sink(t):t\circlearrow\epsilon\\[1ex]
\sw(t,t'):t\cdot t'\circlearrow t'\cdot t
\end{array}
\]
$\id$ corresponds to the identity block.
It copies the input into the output. In the model of constructive functions
$\id(t)$ is the identity function. $\dup(t)$ takes
an input $x$ of type $t$ and outputs $x\cdot x$ of type $t\cdot t$.
$\sink(t)$ returns the empty tuple $\epsilon$, for any input $x$ of type $t$.
$\sw(t,t')$ takes an input
$x\cdot x'$ with $x$ of type $t$ and $x'$ of type $t'$ and returns
$x'\cdot x$. In the model of constructive functions these diagrams
are total functions and they are defined as explained above. In the
abstract model, the behaviors of these constants is defined with a
set of axioms (see below).

\paragraph{Composition operators.}
For two diagrams $S:t\circlearrow t'$ and $S':t'\circlearrow t''$,
their \emph{serial composition}, denoted $S\comp S':t\circlearrow t''$
is a diagram that takes inputs of type $t$ and produces outputs of
type $t''$. In the model of constructive functions, the serial composition
corresponds to function composition ($S\comp S'=S'\circ S$). Please
note that in the abstract model we write the serial composition as
$S\comp S'$, while in the model of constructive functions the first
diagram that is applied to the input occurs second in the composition.

The \emph{parallel composition} of two diagrams $S:t\circlearrow t'$
and $S':r\circlearrow r'$, denoted $S\parallel S':t\cdot r\circlearrow t'\cdot r'$,
is a diagram that takes as input tuples of type $t\cdot r$ and
produces as output tuples of type $t'\cdot r'$. This parallel composition
corresponds to the parallel composition of constructive
functions.

Finally we introduce a \emph{feedback composition}. For $S:a\cdot t\circlearrow a\cdot t'$,
where $a\in\types$ is a single type, the feedback of $S$, denoted
$\fb(S):t\circlearrow t'$, is the result of connecting in feedback
the first output of $S$ to its first input. Again this feedback operation 
corresponds to the feedback of constructive functions.

We assume that parallel composition operator binds stronger than serial
composition, i.e. $S\parallel T\comp R$ is the same as $(S\parallel T)\comp R$.

Graphical diagrams can be represented as terms in the abstract algebra, as illustrated in Figure~\ref{fig:fb-property}. This figure depicts two diagrams, and their corresponding algebra terms. As it turns out, these two diagrams are equivalent, in the sense that their corresponding algebra terms can be shown to be equal using the axioms 
presented below.

\begin{figure}
\centering
\begin{tikzpicture}

\node[draw, minimum height = 5ex, minimum width = 4ex](S){$S$};
\node[draw, minimum height = 5ex, minimum width = 4ex, right = 3ex of S.east, anchor = south west](sw){};
\node[draw, minimum height = 5ex, minimum width = 4ex, right = 11ex of S](T){$T$};

\node[draw,inner xsep = 2ex, fit = (S) (T) (sw)](bx){};

\draw[-latex'](S.south east)++(0, 1ex) coordinate(a) -- (a -| T.west);
\draw[-latex'](S.north east)++(0, -1.5ex) coordinate(Y) -- (Y -| sw.west)coordinate(b);
\draw[-latex'](b) -- ++(4ex,3ex) coordinate(a) -- (a -| bx.east) -- ++(2ex, 0)
    -- ++(0, 3ex) coordinate (b) -- (b -| bx.west) -- ++(-2ex, 0)
    -- ++(0, -3ex) coordinate (c) -- (c -| sw.west) coordinate(X);

\draw[-latex'](X) -- (Y-|sw.east) coordinate(a) -- (a -| T.west);

\draw[-latex'](S.west) ++(-5ex, 0) coordinate(a) -- (S.west);
\draw[latex'-](T.east) ++(5ex, 0) coordinate(a) -- (T.east);


\node[draw, minimum height = 5ex, minimum width = 4ex, right = 45ex of S](S1){$S$};
\node[draw, minimum height = 5ex, minimum width = 4ex, right = 4ex of S1](T1){$T$};

\draw[-latex'](S1.west) ++(-3ex, 0) coordinate(a) -- (S1.west);
\draw[latex'-](T1.east) ++(3ex, 0) coordinate(a) -- (T1.east);

\draw[-latex'](S1.north east)++(0, -1ex) coordinate(a) -- (a -| T1.west);
\draw[-latex'](S1.south east)++(0, 1ex) coordinate(a) -- (a -| T1.west);
\node[below = .1cm of bx] (t1) {$\fb(\id(a)\parallel S\comp\sw(a,a)\parallel\id(t)\comp\id(a)\parallel T)$};
\node[below right = .24cm and -.25cm of S1.south east] (t2) {$S \comp T$};
\node[right = 4.5ex of t1]{$=$};
\end{tikzpicture}
\caption{ \label{fig:fb-property} 
 Two flat diagrams and their corresponding terms in the abstract algebra.}
\end{figure}
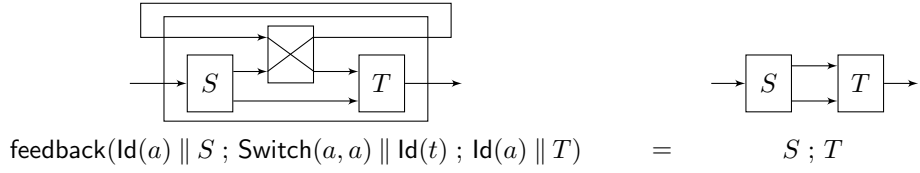

\subsection{Axioms of the Algebra of HBDs}
In the abstract algebra, the behavior of the constants and
composition operators is defined by a set of axioms, listed below:
\begin{enumerate}
\item \label{ax:id-comp} 
	$S:t\circlearrow t'\Longrightarrow\id(t)\comp S=S\comp\id(t')=S$
\\[-1ex]
\item \label{ax:comp-assoc}
$ S:t_{1}\circlearrow t_{2}\land T:t_{2}\circlearrow t_{3}\land R:t_{3}\circlearrow t_{4} \Longrightarrow S\comp(T\comp R)=(S\comp T)\comp R$
\\[-1ex]
\item \label{ax:id-par} 
   $\id(\epsilon)\parallel S=S\parallel\id(\epsilon)=S $
\\[-1ex]
\item \label{ax:par-assoc} 
   $S\parallel(T\parallel R)=(S\parallel T)\parallel R $
\\[-1ex]
\item \label{ax:par-comp} 
$ S:s\circlearrow s'\land S':s'\circlearrow s''\land T:t\circlearrow t'\land T':t'\circlearrow t'' $ \\
\mbox{}$\quad \Longrightarrow\ (S\parallel T)\comp(S'\parallel T')=(S\comp S')\parallel(T\comp T')$
\\[-1ex]
\item \label{ax:split-sink} 
	$\dup(t)\comp\sink(t)\parallel\id(t)=\id(t) $
\\[-1ex]
\item  \label{ax:split-sw} 
	$\dup(t)\comp\sw(t,t)=\dup(t) $
\\[-1ex]
\item \label{ax:split-assoc} 
	$\dup(t)\comp\id(t)\parallel\dup(t)=\dup(t)\comp\dup(t)\parallel\id(t) $
\\[-1ex]
\item \label{ax:sw-cat} 
	$\sw(t,t'\cdot t'')=\sw(t,t')\parallel\id(t'')\comp\id(t')\parallel\sw(t,t'') $
\\[-1ex]
\item \label{ax:sink-cat} 
	$\sink(t\cdot t')=\sink(t)\parallel\sink(t')$
\\[-1ex]
\item \label{ax:split-cat} 
	$\dup(t\cdot t')=\dup(t)\parallel\dup(t')\comp\id(t)\parallel\sw(t,t')\parallel\id(t')$
\\[-1ex]
\item \label{ax:sw-par} 
	$S:s\circlearrow s'\land T:t\circlearrow t'\Longrightarrow\sw(s,t)\comp T\parallel S\comp\sw(t',s')=S\parallel T $
\\[-1ex]
\item \label{ax:fb-sw} 
	$\fb(\sw(a,a))=\id(a) $
\\[-1ex]
\item \label{ax:fb-par} 
	$S:a\cdot s\circlearrow a\cdot t\Longrightarrow\fb(S\parallel T)=\fb(S)\parallel T $
\\[-1ex]
\item \label{ax:fb-comp} 
	$S:a\cdot s\circlearrow a\cdot t\land A:s'\circlearrow s\land B:t\circlearrow t' $ \\ $
\mbox{}\quad \Longrightarrow\ \fb(\id(a)\parallel A\comp S\comp\id(a)\parallel B)=A\comp\fb(S)\comp B$
\\[-1ex]
\item \label{ax:fb-comm} 
	$S:a\cdot b\cdot s\circlearrow a\cdot b\cdot t $\\$
\mbox{}\quad \Longrightarrow\ \fb^{2}(\sw(b,a)\parallel\id(s)\comp S\comp\sw(a,b)\parallel\id(t))=\fb^{2}(S)$
\end{enumerate}
Axioms (\ref{ax:id-comp}) and (\ref{ax:comp-assoc}) express the fact that the identity is the neutral
element for the serial composition, and the serial composition is associative.
Similarly, axioms (\ref{ax:id-par}) and (\ref{ax:par-assoc}) state that the identity of the empty
type is the neutral element for the parallel composition, and that parallel composition
is associative.

Axiom (\ref{ax:par-comp}) introduces a distributivity property of serial and parallel
compositions. Figure \ref{fig:axiom-5} represents graphically this
axiom.
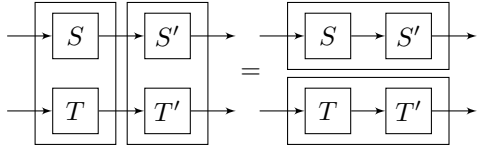
\begin{figure}
\centering
\begin{tikzpicture}
\node[draw,minimum width = 4ex,minimum height = 4ex](a){$S$};
\node[draw, below = 2.5ex of a, minimum width = 4ex,minimum height = 4ex](b){$T$};
\node[draw, inner xsep = 1.5ex, inner ysep = 0.9ex, fit = (a) (b)]{};
\node[draw, right = 4ex of a, minimum width = 4ex,minimum height = 4ex](c){$S'$};
\node[draw, below = 2.5ex of c, minimum width = 4ex,minimum height = 4ex](d){$T'$};
\node[draw, inner xsep = 1.5ex, inner ysep = 0.9ex, fit = (c) (d)](A){};
\draw[-latex'](a)--(c);
\draw[-latex'](b)--(d);
\draw[-latex'](c)--++(6ex,0ex);
\draw[-latex'](d)--++(6ex,0ex);
\draw[latex'-](a)--++(-6ex,0ex);
\draw[latex'-](b)--++(-6ex,0ex);

\node[draw,minimum width = 4ex,minimum height = 4ex, right = 10ex of c](aa){$S$};
\node[draw, below = 2.5ex of aa, minimum width = 4ex,minimum height = 4ex](bb){$T$};
\node[draw, right = 3ex of aa, minimum width = 4ex,minimum height = 4ex](cc){$S'$};
\node[draw, below = 2.5ex of cc, minimum width = 4ex,minimum height = 4ex](dd){$T'$};
\node[draw, inner xsep = 1.5ex, inner ysep = 0.9ex, fit = (aa) (cc)]{};
\node[draw, inner xsep = 1.5ex, inner ysep = 0.9ex, fit = (bb) (dd)]{};

\node[right = 2ex of A]{$=$};

\draw[-latex'](aa)--(cc);
\draw[-latex'](bb)--(dd);
\draw[-latex'](cc)--++(6ex,0ex);
\draw[-latex'](dd)--++(6ex,0ex);
\draw[latex'-](aa)--++(-6ex,0ex);
\draw[latex'-](bb)--++(-6ex,0ex);

\end{tikzpicture}
\caption{Axiom (\ref{ax:par-comp}) {Distributivity of serial and parallel compositions.}
\label{fig:axiom-5}}
\end{figure}

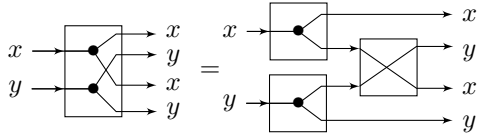
\begin{figure}
\centering
\begin{tikzpicture}
\node[inner ysep = 0ex](a){$x$};
\node[below = 2ex of a, inner ysep = 0ex](b){$y$};
\node[fit=(a) (b)](c){};
\node[draw, minimum width = 5ex,minimum height = 8ex, right = 2ex of c](d){};

\node[inner ysep = 0ex, right = 3ex of d.north east, anchor = north west](e){$x$};
\node[below = 1ex of e, inner ysep = 0ex](f){$y$};
\node[below = 1ex of f, inner ysep = 0ex](g){$x$};
\node[below = 1ex of g, inner ysep = 0ex](h){$y$};

\draw[-latex'](a) -- (a -| d.west);
\draw[-latex'](a-|d.west) -- (a -| d) node[inner sep = 0]{$\bullet$} -- ++(2ex,1.5ex) coordinate(A)-- (A-|e.west);
\draw[-latex'](a) -- (a -| d) -- ++(2ex,-3ex) coordinate(B)-- (B-|g.west);

\draw[-latex'](b) -- (b -| d.west);
\draw[-latex'](b-|d.west) -- (b -| d) node[inner sep = 0]{$\bullet$} -- ++(2ex,3ex) coordinate(C) -- (C-|f.west);
\draw[-latex'](b) -- (b -| d) -- ++(2ex,-2ex) coordinate(D)-- (D-|h.west); 

\node[inner ysep = 0ex, right = 2ex of e](a1){$x$};
\node[below = 5ex of a1, inner ysep = 0ex](b1){$y$};

\node[right = 6ex of d]{$=$};

\node[draw, right = 2ex of a1, minimum width = 5ex, minimum height = 5ex](dx){};

\node[inner ysep = 0.6ex, right = 11ex of dx.north east, anchor = north west](e1){$x$};
\node[below = 1.4ex of e1, inner ysep = 0ex](f1){$y$};
\node[below = 1.8ex of f1, inner ysep = 0ex](g1){$x$};
\node[below = 1.8ex of g1, inner ysep = 0ex](h1){$y$};

\draw[-latex'](a1) -- (a1 -| dx.west);
\draw[-latex'](a1 -| dx.west) -- (a1 -| dx) node[inner sep = 0]{$\bullet$}-- ++(2ex, 1.5ex) coordinate(A) -- (A -| e1.west);
\draw[-latex'](b1) -- (b1 -| dx.west);
\draw[-latex'](b1 -| dx.west) -- (b1 -| dx) node[inner sep = 0]{$\bullet$}-- ++(2ex, -1.5ex) coordinate(A) -- (A -| g1.west);

\node[draw, right = 2ex of b1, minimum width = 5ex, minimum height = 5ex](dy){};

\node[fit=(dx)(dy)](dd){};

\node[draw, right = 2ex of dd, minimum width = 5ex, minimum height = 5ex](sw){};

\draw[-latex'](a1) -- (a1 -| dx) -- ++(2ex, -1.5ex)coordinate(A)--  (A -| sw.west) coordinate(B);
\draw[-latex'](B) -- ++(5ex,-3.5ex) coordinate(A) -- (A -| f1.west);

\draw[-latex'](b1) -- (b1 -| dx) -- ++(2ex, 1.5ex) coordinate(A)--  (A -| sw.west) coordinate(B);
\draw[-latex'](B) -- ++(5ex,3.5ex) coordinate(A) -- (A -| h1.west);
\end{tikzpicture}
\caption{Axiom (\ref{ax:split-cat}) Split switch.\label{fig:axiom-11}}
\end{figure}

Axioms (\ref{ax:split-sink}) -- (\ref{ax:split-cat}) express the properties of $\dup$, $\sink$, and
$\sw$. For example Axiom (\ref{ax:split-cat}), represented in Figure \ref{fig:axiom-11},
says that if we duplicate $x\cdot y$ of type $t\cdot t'$, then this
is equivalent to duplicate $x$ and $y$ in parallel, and then switch
the middle $x$ and $y$.

Axiom (\ref{ax:sw-par}) says that switching the inputs and outputs of $T\parallel S$
is equal to $S\parallel T$.

\begin{figure}
\centering
\begin{tikzpicture}
\node[draw, minimum height = 5ex, minimum width = 4ex](sw){};
\draw[-latex'](sw.south west)++(0, 1ex) coordinate(X) -- ++(4ex, 3ex) -- ++(2ex,0) 
	-- ++(0, 3ex)coordinate (a) -- (a -| sw.west) -- ++ (-2ex, 0) 
    -- ++(0, -3ex) coordinate (b) -- (b -| sw.west) coordinate (c);
\draw[-latex'](c) -- ++(4ex, -3ex) -- ++(3ex, 0) node[anchor = west]{$v$};
\draw[latex'-](X) -- ++(-3ex,0) node[anchor = east]{$u$};

\node[right = 5.5ex of sw]{$=$};
\node[right = 9ex of sw](u){$u$};
\draw[-latex'](u)--+(5ex, 0) node[anchor=west]{$v$};

\end{tikzpicture}
\caption{Axiom (\ref{ax:fb-sw}) Feedback of switch. \label{fig:axiom-13}}
\end{figure}
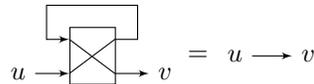

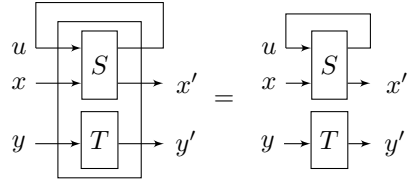
\begin{figure}
\centering
\begin{tikzpicture}
\node[inner ysep = 1ex](u){$u$};
\node[inner ysep = 0ex, below = 1ex of u](x){$x$};
\node[inner ysep = 0ex, below = 4ex of x](y){$y$};
\node[draw, minimum height = 6ex, minimum width = 3ex, right = 4ex of u.north east, anchor=north west](S){$S$};
\node[draw,below = 1ex of S, minimum height = 5ex, minimum width = 3ex](T){$T$};
\node[draw, inner xsep = 2ex,fit=(S)(T)](c){};

\node[inner ysep = 0ex, right = 11.5ex of x](x1){$x'$};
\node[inner ysep = 0ex, right = 11.5ex of y](y1){$y'$};

\draw[-latex'](u.east)coordinate(A) -- (u-|S.west);
\draw [-](u -| S.east) -- ++(4ex,0) -- ++ (0, 4ex) coordinate(B) -- (B -| A) -- (A);

\draw[-latex'](x.east) -- (x-|S.west);
\draw[-latex'](y.east) -- (y-|T.west);

\draw[-latex'](x-|S.east) -- ++(4ex,0);
\draw[-latex'](y-|T.east) -- ++(4ex,0);

\node[right = 5.5ex of c]{$=$};

\node[inner ysep = 1ex, right = 19ex of u](u2){$u$};
\node[inner ysep = 0ex, right = 19ex of x](x2){$x$};
\node[inner ysep = 0ex, right = 19ex of y](y2){$y$};

\node[draw, minimum height = 6ex, minimum width = 3ex, right = 17ex of S](S1){$S$};
\node[draw,below = 1ex of S1, minimum height = 5ex, minimum width = 3ex](T1){$T$};

\draw[-latex'](u2.east)coordinate(A) -- (u2-|S1.west);
\draw [-](u2 -| S1.east) -- ++(2ex,0) -- ++ (0, 3ex) coordinate(B) -- (B -| A) -- (A);

\draw[-latex'](x2.east) -- (x2-|S1.west);
\draw[-latex'](y2.east) -- (y2-|T1.west);

\draw[-latex'](x2-|S1.east) -- ++(2ex,0);
\draw[-latex'](y2-|T1.east) -- ++(2ex,0);

\node[inner ysep = 0ex, right = 8ex of x2](x3){$x'$};
\node[inner ysep = 0ex, right = 8ex of y2](y3){$y'$};

\end{tikzpicture}
\caption{Axiom (\ref{ax:fb-par}) Feedback of parallel.} \label{fig:axiom-14}
\end{figure}

\begin{figure}
\centering
\begin{tikzpicture}
\node[inner ysep = 1ex](u){$u$};
\node[inner ysep = 0ex, below = 2ex of u](x){$x$};

\node[draw, minimum height = 8ex, minimum width = 3ex, right = 9.5ex of u.north east, anchor=north west](S){$S$};
\node[draw, left = 2ex of S.south west, anchor = south east, minimum height = 5ex, minimum width = 3ex](A){$A$};
\node[draw, right = 2ex of S.south east, anchor = south west, minimum height = 5ex, minimum width = 3ex](B){$B$};
\node[draw, inner xsep = 1.5ex, fit = (A) (B) (S)](a){};

\node[inner ysep = 0ex, right = 3.5ex of B](x1){$y$};

\draw[-latex'] (u -| S.east) -- (u -| a.east) -- ++(1.5ex,0) -- ++(0,4ex) coordinate(b) 
    -- (b -| u.east) -- (u.east) -- (u -| S.west);

\draw[-latex'] (x.east) -- (x -| A.west);
\draw[-latex'] (A.east) -- (A -| S.west);
\draw[-latex'] (B.west-|S.east) -- (B.west);
\draw[-latex'] (B.east) -- (x1.west);

\node[right = 11ex of S, anchor = south west]{$=$};

\node[draw, minimum height = 8ex, minimum width = 3ex, right = 25ex of S](S1){$S$};
\node[draw, left = 2ex of S1.south west, anchor = south east, minimum height = 5ex, minimum width = 3ex](A1){$A$};
\node[draw, right = 2ex of S1.south east, anchor = south west, minimum height = 5ex, minimum width = 3ex](B1){$B$};

\node[inner ysep = 0ex, left = 2.5ex of A1](x2){$x$};
\node[inner ysep = 0ex, right = 2.5ex of B1](x3){$y$};

\draw[-latex'](S1.north east) ++(0,-2ex) -- ++ (1.5ex,0) -- ++(0, 3.5ex) 
 coordinate (A) -- (A -| S1.west) -- ++(-2ex, 0) 
    -- ++(0, -3.5ex) coordinate(X) -- (X -| S1.west);

\draw[-latex'] (x2.east) -- (x2 -| A1.west);
\draw[-latex'] (A1.east) -- (A1 -| S1.west);
\draw[-latex'] (B1.west-|S1.east) -- (B1.west);
\draw[-latex'] (B1.east) -- (x3.west);
\end{tikzpicture}
\caption{Axiom (\ref{ax:fb-comp}) Feedback of serial.\label{fig:axiom-15}}
\end{figure}
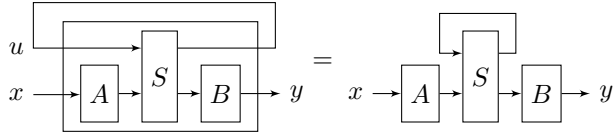

\begin{figure}
\centering
\begin{tikzpicture}

\node[draw, minimum height = 8ex, minimum width = 4ex](S){$S$};
\node[draw, left = 3ex of S.north west, anchor = north east, minimum height = 5ex, minimum width = 4ex](A){};
\node[draw, right = 3ex of S.north east, anchor = north west, minimum height = 5ex, minimum width = 4ex](B){};
\node[draw, inner xsep = 2ex, fit = (A) (B) (S)](aa){};

\draw[-latex'](A.north west)++(0,-1ex)coordinate(X)--++(4ex, -3ex) coordinate(a) -- (a -| S.west);
\draw[-latex'](A.south west)++(0,1ex)coordinate(Y)--++(4ex, 3ex) coordinate(a) -- (a -| S.west);

\draw[-](B.north east)++(0,-1ex)--++(-4ex, -3ex) coordinate(a1);
\draw[-latex'] (a1 -| S.east) -- (a1);
\draw[-](B.south east)++(0,1ex)--++(-4ex, 3ex) coordinate(a2);
\draw[-latex'] (a2 -| S.east) -- (a2);

\draw[-latex'](B.north east)++(0,-1ex) -- ++(4ex, 0) -- ++(0ex, 3ex)coordinate(a) -- (a -| aa.west)
    -- ++(-2ex, 0)coordinate(b) -- (X -| b) -- (X);

\draw[-latex'](B.south east)++(0,1ex) -- ++(5ex, 0) -- ++(0ex, 7ex)coordinate(a) -- (a -| aa.west)
    -- ++(-3ex, 0)coordinate(b) -- (Y -| b) -- (Y);

\draw[-latex'](aa.south west)++(-3ex,2ex) coordinate(a) -- (a -| S.west);
\draw[latex'-](aa.south east)++(3ex,2ex) coordinate(a) -- (a -| S.east);


\node[right  = 14ex of S](eq){$=$};

\node[draw, minimum height = 8ex, minimum width = 4ex, right  = 22ex of S](S1){$S$};
\draw[-latex'](S1.north east)++(0,-1ex) -- ++(2ex, 0) -- ++(0ex, 2ex)coordinate(a) -- (a -| S1.west)
    -- ++(-2ex, 0)coordinate(b) -- ++(0, -2ex) coordinate(b)-- (b -| S1.west);

\draw[-latex'](S1.north east)++(0,-3.8ex) -- ++(3ex, 0) -- ++(0ex, 6ex)coordinate(a) -- (a -| S1.west)
    -- ++(-3ex, 0)coordinate(b) -- ++(0, -5.8ex) coordinate(b)-- (b -| S1.west);

\draw[-latex'](S1.south west)++(-4ex,1.5ex) coordinate(a) -- (a -| S1.west);
\draw[latex'-](S1.south east)++(4ex,1.5ex) coordinate(a) -- (a -| S1.east);
\end{tikzpicture}
\caption{Axiom (\ref{ax:fb-comm}) Feedback of switched inputs/outputs.\label{fig:axiom-16}}

\end{figure}
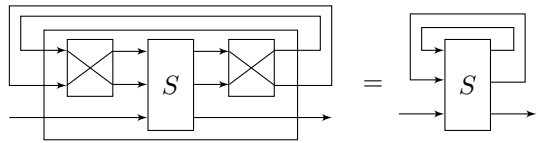

Axioms (\ref{ax:fb-sw}) -- (\ref{ax:fb-comm}) are about the feedback operator. Axiom (\ref{ax:fb-sw}), represented
in Figure \ref{fig:axiom-13}, states that feedback of switch is the identity.
Axiom (\ref{ax:fb-par}), represented in Figure \ref{fig:axiom-14}, states that feedback
of the parallel composition of $S$ and $T$ is the same as the parallel
composition of the feedback of $S$ and $T$. 
Axiom (\ref{ax:fb-comp}), Figure \ref{fig:axiom-15}, states that components $A$
and $B$ can be taken out of the feedback operation. 
Finally, Axiom (\ref{ax:fb-comm}) represented in Figure \ref{fig:axiom-16}, states
that the order in which we apply the feedback operations does not
change the result.

These axioms are equivalent to a subset of the axioms of algebra of flownomials \cite{Stefanescu:2000:NA:518304}, 
which implies that all models of flownomials are also models of our algebra. In \cite{hildebrandt2004}, a relational model for dataflow is introduced. This model is also based on a set of axioms on feedback, serial and parallel compositions, but \cite{hildebrandt2004} does not use the split constant. Our axioms that are not involving split are equivalent to the axioms used in \cite{hildebrandt2004}. The focus of
\cite{hildebrandt2004} is the construction of a relational model for the axioms.

The following theorem provides a concrete semantic domain for HBDs.

\begin{theorem}
Constructive functions with the operations defined in Section \ref{sec:prelim} 
are a model for axioms (\ref{ax:id-comp}) -- (\ref{ax:fb-comm}).
\end{theorem}

We remark that constructive functions are only one example of a model
for axioms (\ref{ax:id-comp}) -- (\ref{ax:fb-comm}), and by no means
the only model. As mentioned above, 
all models of flownomials are also models of our algebra.
In particular, relations are a model of flownomials and therefore
also a model for axioms (\ref{ax:id-comp}) -- (\ref{ax:fb-comm})
\cite{Stefanescu:2000:NA:518304}.

\section{The Abstract Translation Algorithm and its Determinacy}
\label{sec:algorithm}

\subsection{Diagrams with Named Inputs and Outputs}
\label{subsec:diagram}

The algorithm works by first transforming the graph of a HBD into
a list of basic components with named inputs and outputs as explained
in Section \ref{sec:overview}. For this purpose we
assume a set of names or variables $\var$ and a function 
$\tv:\var\to\types$.
For $v\in\var$, $\tv(v)$ is the type of variable $v$. We extend
$\tv$ to lists of variables by 
$\tv(v_1,\ldots, v_n) = (T(v_1),\ldots,T(v_n))$.
\begin{definition}
A \emph{diagram with named inputs and outputs} or \emph{io-diagram}
for short is a tuple $(\mathit{in},\mathit{out}, S)$ such that $\mathit{in},\mathit{out}\in\var^{*}$
are lists of distinct variables, and $S:\tv(\mathit{in})\circlearrow\tv(\mathit{out})$.
\end{definition}

In what follows we use the symbols $A,A',B,\ldots$ to denote io-diagrams,
and $\In(A)$, $\Out(A)$, and $\dg(A)$ to denote the input variables,
the output variables, and the diagram of $A$, respectively.

\begin{definition}
For io-diagrams $A$ and $B$, we define $\Var(A,B)=\Out(A)\linter\In(B)\in\var^{*}$.
\end{definition}
$\Var(A,B)$ is the list of common variables that are output of $A$ and input
of $B$, in the order occurring in $\Out(A)$. We use $\Var(A,B)$ later to
connect for example in series $A$ and $B$ on these common variables, as we did
for constructing $A$ from $\add$ and $\delay$ in Section~\ref{sec:overview}.

\subsection{General Switch Diagrams}
\label{subsec:switch}

We compose diagrams when their types are matching, and we compose
io-diagrams based on matching names of input output variables. For
example if we have two io-diagrams $A$ and $B$ with $\Out(A)=u\cdot v$
and $\In(B)=v\cdot u$, then we can compose in series $A$ and $B$
by switching the output of $A$ and feeding it into $B$, i.e.,  ($A\comp\sw(\tv(u),\tv(v))\comp B$). 

In general, for two lists of variables $x=(x_{1}\cdots x_{n})$ and $y=(y_{1}\cdots y_{k})$
we define a \emph{general switch diagram} $[x_{1}\cdots x_{n}\leadsto y_{1}\cdots y_{k}]:\tv(x_{1}\cdots x_{n})\circlearrow\tv(y_{1}\cdots y_{k})$.
Intuitively this diagram takes as input a list of values of type $\tv(x_{1}\cdots x_{n})$
and outputs a list of values of type $\tv(y_{1}\cdots y_{k})$, where
the output value corresponding to variable $y_{j}$ is equal to the
value corresponding to the first $x_{i}$ with $x_{i}=y_{j}$ and it
is arbitrary (unknown) if there is no such $x_{i}$. For example in
the constructive functions model $[u,v\leadsto v,u,w,u]$ for input $(a, b)$
outputs $(b,a,\bot,a)$. 

To define $[\_\leadsto\_]$ we use $\dup$, $\sink$, and $\sw$,
but we need also an additional diagram that outputs an arbitrary (or
unknown) value for an empty input. For $a\in\types$, we define $\arb(a):\epsilon\circlearrow a$
by 
\[
\arb(a)=\fb(\dup(a))
\]
The diagram $\arb$ is represented in Figure \ref{fig:arb}.

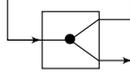
\begin{figure}
\centering
\begin{tikzpicture}
\node[draw, minimum height = 5ex, minimum width = 5ex](sw){};
\draw[-latex'](sw.west)++(-3ex, 0ex) coordinate(X) -- ++(5.5ex,0) node{$\bullet$} coordinate(Z) -- ++(2.5ex, 1.8ex) --
	++(3ex,0) -- ++(0, 2.2ex) coordinate (Y) -- (Y -| X) -- (X) -- (sw.west);
\draw[-latex'] (Z) -- ++(2.5ex, -1.8ex) -- ++(3ex,0);
\end{tikzpicture}
\caption{The diagram $\arb$. \label{fig:arb}}
\end{figure}

We define now $[x\leadsto y]:\tv(x)\circlearrow\tv(y)$ in two steps.
First for $x\in\var^{*}$ and $u\in\var$, the diagram $[x\leadsto u]:\tv(x)\circlearrow\tv(u)$,
for input $a_{1},\ldots,a_{n}$ it outputs the value
$a_{i}$ where $i$ is the first index such that $x_{i}=u$. Otherwise
it outputs an arbitrary (unknown) value. 
\[
\begin{array}{lllc}
[\epsilon\leadsto u] & = & \arb(\tv(u))\\[1ex]
[u\cdot x\leadsto u] & = & \id(\tv(u))\parallel\sink(\tv(x))\\[1ex]
[v\cdot x\leadsto u] & = & \sink(\tv(v))\parallel[x\leadsto u] & \ (\mbox{if}\ u\not=v)
\end{array}
\]
and
\[
\begin{array}{lll}
[x\leadsto\epsilon] & = & \sink(\tv(x))\\[1ex]
[x\leadsto u\cdot y] & = & \dup(\tv(x))\comp([x\leadsto u]\parallel[x\leadsto y])
\end{array}
\]

\subsection{Basic Operations of the Abstract Translation Algorithm}
\label{subsec:operations}

The algorithm starts with a list of io-diagrams and repeatedly
applies operations until it reduces the list to only one io-diagram.
These operations are the extensions of serial, parallel and feedback
from diagrams to io-diagrams.
\begin{definition}
The named serial composition of two io-diagrams $A$ and $B$, denoted
$A\Comp B$ is defined by $A\Comp B=(\mathit{in},\mathit{out}, S)$,
where $x=\In(B)\ominus\Var(A,B)$, $y=\Out(A)\ominus \Var(A,B)$, $\mathit{in}=\In(A)\lunion x$,
$\mathit{out=y\cdot\Out(B)}$ and 
\[
S=[\mathit{in}\leadsto\In(A)\cdot x]\comp\dg(A)\parallel[x\leadsto x]\comp[\Out(A)\cdot x\leadsto y\cdot\In(B)]\comp[y\leadsto y]\parallel\dg(B)
\]
\end{definition}
The construction of $A$ from Section \ref{sec:overview} can be obtained by
applying the named serial composition to $\add$ and $\delay$. 

Figure \ref{fig:named-serial} illustrates an example of the named
serial composition. In this case we have $\Var(A,B)=u$, $x=(a,b)$,
$y=(v,w)$, $\mathit{in}=(a,c,b)$, and $\mathit{out}=(v,w,d,e)$.
The component $A$ has outputs $u,v,w$, and $u$ is also input of
$B$. Variable $u$ is the only variable that is output of $A$ and
input of $B$. Because the outputs $v,w$ of $A$ are not inputs of
$B$ they become outputs of $A\Comp B$. Variable $a$ is input for
both $A$ and $B$, so in $A\Comp B$ the value of $a$ is split and
fed into both $A$ and $B$. The diagram for this example is:
$$
[a,c,b\leadsto a,c,a,b]\comp A \parallel \id(\tv(a,b))\comp 
 [u,v,w,a,b\leadsto v,w,a,u,b] \comp \id(\tv(v,w))\parallel B.
$$

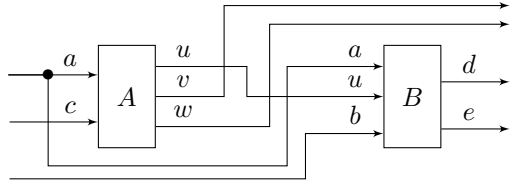
\begin{figure}
\begin{centering}
\begin{tikzpicture}
\node[draw, minimum height = 9ex, minimum width = 5ex](A){$A$};
\node[right = 1ex of A.east, anchor = south west](v){$v$};
\node[above = 0ex of v](u){$u$};
\node[below = 0ex of v](w){$w$};

\node[draw, minimum height = 9ex, minimum width = 5ex, right = 20ex of A](B){$B$};
\node[left = 1ex of B.west, anchor = south east](u1){$u$};

\node [left = 1ex of A.north west, anchor = north east](u2){$a$};
\node [below = 1.5ex of u2](v1){$c$};

\node[above = 0ex of u1](u3){$a$};
\node[below = 0ex of u1](w1){$b$};

\node [right = 1ex of B.north east, anchor = north west](d){$d$};
\node [below = 1.5ex of d](e){$e$};

\draw[-latex'](d.south west -| B.east) -- ++(6ex,0) coordinate (U);
\draw[-latex'](e.south west -| B.east) -- ++(6ex,0);

\draw[-latex'] (u.south west -| A.east) -- ++(8ex,0) coordinate (X) -- (u1.south west -| X)
-- (u1.south west -| B.west);

\draw[-latex'] (v.south west -| A.east) -- ++(6ex,0) -- ++(0ex, 8ex) coordinate (T) 
	-- (T -| U);
\draw[-latex'] (w.south west -| A.east) -- ++(10ex,0) -- ++(0ex, 9ex) coordinate (T) 
	-- (T -| U);

\draw[-latex'](u2.south west) ++ (-4ex, 0) coordinate (X) -- (X -| A.west);
\draw[-latex'] (X) -- ++(3.5ex,0)node{$\bullet$} -- ++(0,-8ex) -- ++(21ex,0) coordinate (Y) 
	-- (u3.south-|Y) -- (u3.south -| B.west);
\draw[-latex'](v1.south west) ++ (-4ex, 0) -- (v1.south west -| A.west);

\draw[-latex'](v1.south west) ++ (-4ex, -5ex) -- ++(26ex,0) coordinate (X)
	-- (w1.south-|X) -- (w1.south -| B.west);
\end{tikzpicture}
\par\end{centering}
\caption{Example of named serial composition. \label{fig:named-serial}}
\end{figure}

The result of the named serial composition of two io-diagrams is not
always an io-diagram. The problem is that the outputs of $A\Comp B$
are not distinct in general. The next lemma gives sufficient conditions
for $A\Comp B$ to be an io-diagram.
\begin{lemma}
If $A,B$ are io-diagrams and $(\Out(A)\ominus\In(B))\linter\Out(B)=\epsilon$
then $A\Comp B$ is an io-diagram. In particular if $\Out(A)\linter\Out(B)=\epsilon$
then $A\Comp B$ is an io-diagram.
\end{lemma}
The named serial composition is associative, expressed by the next lemma.
\begin{lemma}
If $A,B,C$ are io-diagrams such that $(\Out(A)\ominus\In(B))\linter\Out(B)=\epsilon$
and $(\Out(A) \linter \In(B))\linter \In(C) = \epsilon$ then
$$
(A \Comp B) \Comp C = A \Comp (B \Comp C)
$$
\end{lemma}

Next we introduce the corresponding operation on io-diagrams for the
parallel composition.
\begin{definition}
If $A,B$ are io-diagrams, then the named parallel composition of
$A$ and $B$, denoted $A\Parallel B$ is defined by 
$$A\Parallel B = (\In(A)\lunion \In(B),\Out(A)\cdot\Out(B),S)$$
where
\[
S=[\In(A)\lunion\In(B)\leadsto\In(A)\cdot\In(B)]\comp(A\parallel B)
\]
\end{definition}
Figure~\ref{fig:named-parallel} presents an example of a named parallel composition.
The named parallel composition is meaningful only if the outputs
of the two diagrams have different names. However, the inputs may not necessarily be distinct
as shown in Figure~\ref{fig:named-parallel}.

\begin{figure}
\begin{centering}
\begin{tikzpicture}
\node[draw, minimum height = 9ex, minimum width = 5ex](A){$A$};
\node[right = 1ex of A.east, anchor = south west](v){$v$};
\node[above = 0ex of v](u){$u$};
\node[below = 0ex of v](w){$w$};

\node[left = 1ex of A.west, anchor = south east](b){$b$};
\node[above = -0.3ex of b](a){$a$};
\node[below = 0ex of b](c){$c$};

\node[draw, minimum height = 9ex, minimum width = 5ex, below = 2ex of A](B){$B$};
\node[right = 1ex of B.east, anchor = south west](s){$s$};
\node[above = 0ex of s](t){$t$};
\node[below = 0ex of s](r){$r$};

\node[left = 1ex of B.west, anchor = south east](ba){$b$};
\node[above = -0.3ex of ba](d){$d$};
\node[below = 0ex of ba](aa){$a$};

\draw[-latex'](u.south west -| A.east) -- ++(6ex,0);
\draw[-latex'](v.south west -| A.east) -- ++(6ex,0);
\draw[-latex'](w.south west -| A.east) -- ++(6ex,0);

\draw[-latex'](s.south west -| B.east) -- ++(6ex,0);
\draw[-latex'](t.south west -| B.east) -- ++(6ex,0);
\draw[-latex'](r.south west -| B.east) -- ++(6ex,0);

\draw[latex'-](a.south west -| A.west) -- ++(-10ex,0) coordinate(X);
\draw[latex'-](b.south west -| A.west) -- ++(-10ex,0) coordinate(Z);
\draw[latex'-](c.south west -| A.west) -- ++(-10ex,0);

\draw[latex'-](d.south west -| B.west) -- ++(-10ex,0);
\draw[-latex'](X) ++(2ex,0)node[inner sep = 0](Y){$\bullet$} -- (aa.south -| Y) -- 
	(aa.south -| B.west);

\draw[-latex'](Z) ++(4ex,0)node[inner sep = 0](Y){$\bullet$} -- (ba.south -| Y) -- 
	(ba.south -| B.west);

\end{tikzpicture}
\par\end{centering}
\caption{Example of named parallel composition. \label{fig:named-parallel}}
\end{figure}
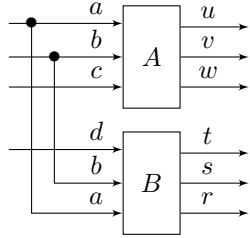

As in the case of named serial composition, the parallel composition of two 
io-diagrams is not always an io-diagram. 
Next lemma gives conditions for
the parallel composition to be io-diagram and also states that  
the named parallel composition is associative.
\begin{lemma} Let $A$, $B$, and $C$ be io-diagrams, then
\begin{enumerate}
\item
$\Out(A)\linter\Out(B)=\epsilon \ \Rightarrow \ A \Parallel B$ is an io-diagram.\\[-1.5ex]
\item $(A\Parallel B)\Parallel C=A\Parallel(B\Parallel C)$
\end{enumerate}
\end{lemma}

Next definition introduces the feedback operator for io-diagrams.
\begin{definition}
If $A$ is an io-diagram, then the named feedback of $A$, denoted
$\FB(A)$ is defined by $(\mathit{in},\mathit{out},S)$, where $\mathit{in}=\In(A)\ominus\Var(A,A)$,
$\mathit{out}=\Out(A)\ominus\Var(A,A)$ and 
\[
S=\fb^{|\Var(A,A)|}([\Var(A,A)\cdot in\leadsto\In(A)]\comp S\comp[\Out(A)\leadsto\Var(A,A)\cdot\mathit{out}])
\]
\end{definition}
The named feedback operation of $A$ connects all inputs and outputs of
$A$ with the same name in feedback. 
Figure~\ref{fig:named-feedback} illustrates an example of named feedback composition.
The named feedback applied to
an io-diagram is always an io-diagram.
\begin{lemma}
If $A$ is an io-diagram then $\FB(A)$ is an io-diagram.
\end{lemma}

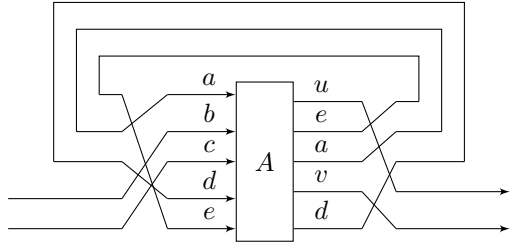
\begin{figure}
\begin{centering}
\begin{tikzpicture}
\node[draw, minimum height = 14ex, minimum width = 5ex](main){$A$};
\node[right = 1ex of main.east, anchor = south west](aa){$a$};
\node[above = 0ex of aa](ee){$e$};
\node[above = 0ex of ee](u){$u$};
\node[below = 0ex of aa](v){$v$};
\node[below = 0ex of v](dd){$d$};

\node[left = 1ex of main.west, anchor = south east](c){$c$};
\node[above = 0ex of c](b){$b$};
\node[above = 0ex of b](a){$a$};
\node[below = 0ex of c](d){$d$};
\node[below = 0ex of d](e){$e$};

\draw[-](u.south west -| main.east) -- ++(6ex,0) coordinate(U);
\draw[-](ee.south west -| main.east) -- ++(6ex,0) coordinate(EE);
\draw[-](aa.south west -| main.east) -- ++(6ex,0) coordinate(AA);
\draw[-](v.south west -| main.east) -- ++(6ex,0) coordinate(V);
\draw[-](dd.south west -| main.east) -- ++(6ex,0) coordinate(DD);

\draw[-](U) ++(3ex, 0) coordinate (EEa) -- ++(2ex, 0) coordinate(EEb);
\draw[-](EE) -- (EEa); 

\draw[-](EE) ++(3ex, 0) coordinate (AAa) -- ++(4ex, 0) coordinate(AAb);
\draw[-](AA) -- (AAa); 

\draw[-](AA) ++(3ex, 0) coordinate (DDa) -- ++(6ex, 0) coordinate(DDb);
\draw[-](DD) -- (DDa); 

\draw[-latex'](V) ++(3ex, 0) coordinate (Ua) -- ++(10ex, 0) coordinate(Ub);
\draw[-](U) -- (Ua); 

\draw[-latex'](DD) ++(3ex, 0) coordinate (Va) -- ++(10ex, 0) coordinate(Vb);
\draw[-](V) -- (Va);

\draw[latex'-](a.south west -| main.west) -- ++(-6ex,0) coordinate(A);
\draw[latex'-](b.south west -| main.west) -- ++(-6ex,0) coordinate(B);
\draw[latex'-](c.south west -| main.west) -- ++(-6ex,0) coordinate(C);
\draw[latex'-](d.south west -| main.west) -- ++(-6ex,0) coordinate(D);
\draw[latex'-](e.south west -| main.west) -- ++(-6ex,0) coordinate(E);

\draw[-](A) ++(-4ex, 0) coordinate (Ea) -- ++(-2ex, 0) coordinate(Eb);
\draw[-](E) -- (Ea); 

\draw[-](B) ++(-4ex, 0) coordinate (Aa) -- ++(-4ex, 0) coordinate(Ab);
\draw[-](A) -- (Aa); 

\draw[-](C) ++(-4ex, 0) coordinate (Da) -- ++(-6ex, 0) coordinate(Db);
\draw[-](D) -- (Da); 

\draw[-](D) ++(-4ex, 0) coordinate (Ba) -- ++(-10ex, 0) coordinate(Bb);
\draw[-](B) -- (Ba); 

\draw[-](E) ++(-4ex, 0) coordinate (Ca) -- ++(-10ex, 0) coordinate(Cb);
\draw[-](C) -- (Ca); 

\draw[-](EEb) -- ++(0,4ex) coordinate(X) -- (X -| Eb) -- (Eb);
\draw[-](AAb) -- ++(0,9ex) coordinate(X) -- (X -| Ab) -- (Ab);
\draw[-](DDb) -- ++(0,14ex) coordinate(X) -- (X -| Db) -- (Db);

\end{tikzpicture}
\par\end{centering}
\caption{Example of named feedback composition. \label{fig:named-feedback}}
\end{figure}

\subsection{The Abstract Translation Algorithm}
We have now all elements for introducing the abstract translation algorithm.
The algorithm starts with a list ${\cal A}=(A_{1},A_{2},\ldots,A_{n})$
of io-diagrams, such that for all $i\not=j$, the inputs and outputs
of $A_{i}$ and $A_{j}$ are disjoint respectively ($\In(A_{i})\linter\In(A_{j})=\epsilon$
and $\Out(A_{i})\linter\Out(A_{j})=\epsilon$). We denote this property
by $\iodist({\cal A})$. The algorithm is given in Alg.~\ref{alg:algorithm}.
Formally the algorithm is represented as a monotonic predicate transformer
\cite{dijkstra:1975}, within the framework of refinement calculus 
\cite{back-wright-98}.

\begin{table}[H]
\renewcommand{\tablename}{\bf Alg.}
\centering
\normalsize
\begin{minipage}{0.9\textwidth}
{\sf
input: ${\cal A} = (A_{1},A_{2},\ldots,A_{n})$ \ \ \ {\rm (list of io-diagrams)}\\[1ex]
$\mathsf{while}\ |{\cal A}|>1:$\\[1ex]
\mbox{}~~~~$\mathsf{choose}$:\\[1ex]
\mbox{}~~~~~~~~(a) $[\,{\cal A} := {\cal A'}\  |\ \exists\ k, B_1,\ldots, B_k, {\cal C}: k > 1 
    \land \perm({\cal A},\, (B_1, \ldots , B_k)\cdot {\cal C}) $ \\[1ex]
\mbox{}~~~~~~~~~~~~~~~~$\land~~{\cal A}' = \FB(B_{1}\Parallel\ldots\Parallel B_{k}) \cdot {\cal C}\,]$
    \\[1ex]
\mbox{}~~~~~~~~(b) $[\,{\cal A} := {\cal A'}\  |\ \exists\ A,B, {\cal C}: \perm({\cal A},\, 
     (A, B) \cdot {\cal C}) \land~{\cal A}' = \FB(\FB(A)\Comp \FB(B)) \cdot {\cal C}\,]$
     \\[1ex]
$A := \FB(A')$ \  \ \rm (where $A'$ is the only remaining element of ${\cal A}$)
}
\end{minipage}
\\[2ex]
\caption{Nondeterministic algorithm for translating HBDs.}
\label{alg:algorithm}
\end{table}
Computing $\FB(A)$ in the last step of the algorithm is necessary
only if ${\cal A}$ contains initially only one element. However,
computing $\FB(A)$ always at the end does not change the result
since,
as we will see later in Theorem \ref{thm:par-ser}, $\FB$ operation 
is idempotent, i.e. $\FB(\FB(A)) = \FB(A)$. In the presentation of the algorithm,
we have used the keyword {\sf choose} for the nondeterministic choice $\sqcap$, 
to emphasize the two alternatives.

Note that, semantically, choice {\sf (b)} of the algorithm is a special case of choice {\sf (a)}, as shown later in Theorem \ref{thm:par-ser}. But syntactically, choices {\sf (a)} and {\sf (b)} result in different expressions that achieve different performance tradeoffs as observed in Section~\ref{sec:overview} and
as further discussed in \cite{DBLP:conf/spin/DragomirPT16}.
The point of our translator is to be indeed able to generate semantically equivalent but syntactically different expressions, which achieve different performance tradeoffs~\cite{DBLP:conf/spin/DragomirPT16}.


The result for the running example from Section \ref{sec:overview} can be obtained
by applying the second choice of the algorithm twice for the initial list of io-diagrams
($[\add,\delay,\split]$), first to $\add$ and $\delay$ to obtain $A$,
and next to $A$ and $\split$ to obtain
\[
\Big((u,s),\ (v,s'),\ \fb\big((\dg(\add)\parallel\id)\comp\dg(\delay)\comp ((\split)
\parallel\id)\big)\Big).
\]
As opposed to the example from Section~\ref{sec:overview}, the elements are composed
serially in the order occurring in the diagram.

\subsection{Determinacy of the Abstract Translation Algorithm}

The result of the algorithm depends on how the nondeterministic choices are resolved.
However, in all cases the final io-diagrams are equivalent modulo a
permutation of the inputs and outputs. To prove this, we introduce
the concept \emph{io-equivalence} for two io-diagrams.
\begin{definition}
Two io-diagrams $A,B$ are \emph{io-equivalent}, denoted $A\ioeq B$
if they are equal modulo a permutation of the inputs and outputs,
i.e., $\In(B)$ is a permutation of $\In(A)$, $\Out(B)$ is a permutation of $\Out(A)$ and 
\[
\dg(A)=[\In(A)\leadsto\In(B)]\comp\dg(B)\comp[\Out(B)\leadsto\Out(A)]
\]
\end{definition}

\begin{lemma}\label{lem:ioeq-cong}
The relation io-equivalent is a congruence relation, i.e, for all
$A,B,C$ io-diagrams we have:
\begin{enumerate}
\item $A\ioeq A$\\[-1.5ex]
\item $A\ioeq B\Rightarrow B\ioeq A$\\[-1.5ex]
\item $A\ioeq B\land B\ioeq C\Rightarrow A\ioeq C$.\\[-1.5ex]
\item $A\ioeq B \Rightarrow \FB(A)\ioeq\FB(B)$.\\[-1.5ex]
\item $\Out(A) \linter \Out(B) =\epsilon \Rightarrow  A\Parallel B\ioeq B\Parallel A$.\\[-1.5ex]
\item If $\iodist(A_1,\ldots,A_n)$ and $\perm((A_1,\ldots,A_n),\ (B_1,\ldots, B_n))$
 then $$A_1\Parallel\ldots A_n \ioeq B_1\Parallel\ldots B_n.$$
\end{enumerate}
\end{lemma}

To prove correctness of the algorithm we also need the following results:
\begin{theorem}
\label{thm:par-ser}If $A,B$ are io-diagrams such that $\In(A)\linter\In(B)=\epsilon$
and $\Out(A)\linter\Out(B)=\epsilon$ then 
\[
\FB(A\Parallel B)=\FB(\FB(A)\Comp\FB(B))
\]
and
\[
\FB(\FB(A))=\FB(A).
\]
\end{theorem}
The proof of Theorem~\ref{thm:par-ser} is quite involved and requires several properties of diagrams (see the RCRS formalization~\cite{RCRS_Toolset_figshare} for details).

We can now state and prove one of the main results of this paper,
namely, determinacy of Algorithm~\ref{alg:algorithm}.
\begin{theorem}\label{thm:correct}
$If$ ${\cal A}=(A_{1},A_{2},\ldots,A_{n})$ is the initial list of
io-diagrams satisfying $\iodist({\cal A})$, then Algorithm~\ref{alg:algorithm} 
terminates, and if $A$ is the io-diagram computed by the algorithm,
then
\[
A\ioeq\FB(A_{1}\Parallel\ldots\Parallel A_{n})
\]
\end{theorem}

\begin{proof}
It is easy to see that the algorithm terminates because at each step,
the size of the list ${\cal A}$ decreases. The termination variant
in Lemma~\ref{hoare:while} is $|{\cal A}|$, the length of list ${\cal A}$. 

To prove the correctness
of the algorithm we use the Hoare rule for the while
statement (Lemma~\ref{hoare:while}), which requires an invariant.
The invariant must be true at the beginning of the while loop, it
must be preserved by the body of the while loop, and it must establish
the final post-condition ($A\ioeq\FB(A_{1}\Parallel\ldots\Parallel A_{n})$).
If ${\cal A}_{0}=(A_{1},\ldots,A_{n})$ is the initial list of the
io-diagrams, and ${\cal A}=(C_{1},\ldots,C_{m})$ is the current list
of io-diagrams, then the invariant is
\[
inv({\cal A})=\iodist({\cal A})\land\FB(C_{1}\Parallel\ldots\Parallel C_{m})\ioeq\FB(A_{1}\Parallel\ldots\Parallel A_{n})
\]
Initially $inv({\cal A})$ is trivially true, and it also trivially
establishes the final post-condition. We need to prove that both choices
in the algorithm preserve the invariant.

\begin{equation}
\begin{array}{l}
inv({\cal A})\land k>1 \land \perm({\cal A},\,(B_1,\ldots, B_k)\cdot {\cal C})
\Rightarrow inv([\FB(B_{1}\Parallel\ldots\Parallel B_{k})]\cdot {\cal C})
\end{array}\label{eq:inv1}
\end{equation}
and
\begin{equation}
\begin{array}{l}
inv({\cal A})\land \perm({\cal A},\,(A, B)\cdot {\cal C}) \Rightarrow inv([\FB(\FB(A) \Comp\FB(B))]\cdot {\cal C})
\end{array}\label{eq:inv2}
\end{equation}

The properties (\ref{eq:inv1}) and (\ref{eq:inv2}) are obtained by applying the 
Hoare rule for the nondeterministic choice, and then the rule for nondeterministic
assignment (Lemma~\ref{hoare:assign}).

We prove (\ref{eq:inv1}). Assume 
$${\cal A}=(C_{1},\ldots,C_{m})
\mbox{\ \ and \ \ }
inv({\cal A})=\iodist({\cal A})\land\FB(C_{1}\Parallel\ldots\Parallel C_{m})\ioeq\FB(A_{1}\Parallel\ldots\Parallel A_{n}).
$$
Let $D_{1}=\FB(B_{1}\Parallel\ldots\Parallel B_{k})$, and ${\cal C}=(D_{2},\ldots,D_{u})$.
It follows that $\iodist(D_{1},\ldots,D_{u})$. We prove now
that $\FB(D_{1}\Parallel\ldots\Parallel D_{u})\ioeq \FB(A_{1}\Parallel\ldots\Parallel A_{n})$.
\begin{lyxlist}{00.0}
\item [] $\FB(D_{1}\Parallel\ldots\Parallel D_{u})$\\[-1.5ex]
\item [{$=$}] \{Theorem \ref{thm:par-ser} and $\Parallel$ is associative\}\\[-1.5ex]
\item [{~}] $\FB(\FB(D_{1})\Comp\FB(D_{2}\Parallel\ldots\Parallel D_{u}))$\\[-1.5ex]
\item [{$=$}] \{Definition of $D_{1}$\}\\[-1.ex]
\item [{~}] $\FB(\FB(\FB(B_{1}\Parallel\ldots\Parallel B_{k}))\Comp\FB(D_{2}\Parallel\ldots\Parallel D_{u}))$\\[-1.5ex]
\item [{$=$}] \{Theorem \ref{thm:par-ser}\}\\[-1.5ex]
\item [{~}] $\FB(\FB(B_{1}\Parallel\ldots\Parallel B_{k})\Comp\FB(D_{2}\Parallel\ldots\Parallel D_{u}))$\\[-1.5ex]
\item [{$=$}] \{Theorem \ref{thm:par-ser} and $\Parallel$ is associative\}\\[-1.5ex]
\item [{~}] $\FB(B_{1}\Parallel\ldots\Parallel B_{k}\Parallel D_{2}\Parallel\ldots\Parallel D_{u})$
\\[-1.5ex]
\item [{$\ioeq$}] \{Lemma \ref{lem:ioeq-cong} and $\perm((B_{1},\ldots,B_{k},D_{2},
\ldots,D_{u}),\,(C_{1},\ldots,C_{m}))$\}\\[-1.5ex]
\item [{~}] $\FB(C_{1}\Parallel\ldots\Parallel C_{m})$\\[-1.5ex]
\item [{$\ioeq$}] \{Assumptions\}\\[-1.5ex]
\item [{~}] $\FB(A_{1}\Parallel\ldots\Parallel A_{n})$
\end{lyxlist}
Property (\ref{eq:inv2}) can be reduced to property (\ref{eq:inv1})
by applying Theorem \ref{thm:par-ser}.
\end{proof}

\section{Proving Equivalence of Two Translation Strategies}
\label{sec:impl}

To demonstrate the usefulness of our framework, we return to our 
original motivation, namely, the open problem of
how to prove equivalence of the translation 
strategies introduced in \cite{DBLP:conf/spin/DragomirPT16}.
Two of the translation strategies of~\cite{DBLP:conf/spin/DragomirPT16},
called \emph{feedback-parallel} and \emph{incremental}
translation, can be seen as a determinizations (or refinements)
of the abstract algorithm of Section~\ref{sec:algorithm}, and therefore
can be shown to be equivalent and correct with respect to the abstract semantics.
(The third strategy proposed in \cite{DBLP:conf/spin/DragomirPT16}, called {\em feedbackless}, is significantly different and is presented in the next section.)

The feedback-parallel strategy is the implementation of the abstract
algorithm where we choose $k=\vert\mathcal{A}\vert$. Intuitively, all diagram
components are put in parallel and the common inputs and outputs are connected via $\fb$ operators. On the running example from Figure
\ref{fig:sum-state}, this strategy will generate the following component:

\[
\begin{array}{l}
((u,s),\ (v,s'),\ \fb^{3}([z, x, y, u, s\leadsto z, u, x, s, y]\\[1ex]
\quad \comp \dg(\add) \parallel \dg(\delay) \parallel 
\dg(\split)
 \comp[x, y, s', z, v\leadsto z, x, y, v, s']))
\end{array}
\]
The switches are ordering the variables such that the feedback
variables are first and in the same order in both input and output lists.

The incremental strategy is the implementation of the abstract algorithm
where we use only the second choice of the algorithm and the first two components of the list $\cal A$. This strategy is dependent on the
initial order of $\cal A$, and we order $\cal A$ topologically (based on the input - output connections) at the beginning, in 
order to reduce the number of switches needed.  

Again on the running example, assume
that this strategy composes first $\add$ with $\delay$, and the
result is composed with $\split$. The following component is then
obtained:

\[
((u,s),\,(v,s'),\,\fb(\dg(\add)\parallel\id\comp\dg(\delay)\comp\dg(\split)
\parallel\id)
\]

The $\add$ and $\split$ components are put in parallel with $\id$
for the unconnected input and output state respectively. Next all
components are connected in series with one $\fb$ operator for the
variable $z$.

The next theorem shows that the two strategies are equivalent, and that they are independent of the initial order of $\cal A$.
\begin{theorem}
If $A$ and $B$ are the result of the feedback-parallel and incremental strategies on $\cal A$, respectively, then $A$ and $B$ are input - output equivalent 
($A\ioeq B$). Moreover both strategies are independent of the initial 
order of $\cal A$.
\end{theorem}
\proof 
Both strategies are refinements of the
nondeterministic algorithm. Therefore, using Lemma~\ref{hoare:refinement}, 
they satisfy the same
correctness conditions (Theorem \ref{thm:correct}), i.e.
$$
A\ioeq \FB(A_1\Parallel \ldots \Parallel A_n) \mbox{ and }
B\ioeq \FB(A_1\Parallel \ldots \Parallel A_n)
$$
where ${\cal A} = (A_1,\ldots,A_n)$.
From this, since $\ioeq$ is transitive and symmetric, we obtain 
 $A\ioeq B$.
 
For the second part, we use a similar reasoning. Let ${\cal A}
  = [A_1,\ldots,A_n]$, and ${\cal B} = [B_1,\ldots,B_n]$ a 
  permutation of $\cal A$.
If $A$ and $B$ are the outputs of feedback-parallel on $\cal A$ and
$\cal B$, respectively, then we prove $A\ioeq B$. Using Theorem \ref{thm:correct}
again we have:\\[-2ex]
$$
A \ioeq \FB(A_1\Parallel \ldots \Parallel A_n)  \mbox{ and }
B \ioeq \FB(B_1\Parallel \ldots \Parallel B_n).
$$
Moreover, because ${\cal B}$ is a permutation of $\cal A$, using Lemma
\ref{lem:ioeq-cong} 
we have
$$\FB(A_1\Parallel \ldots \Parallel A_n) \ioeq \FB(B_1\Parallel \ldots \Parallel B_n).$$
Therefore $A\ioeq B$. The same holds for the incremental strategy. 
\qed

Since both strategies are refinements of the nondeterministic algorithm,
they both satisfy the same correctness conditions of Theorem
\ref{thm:correct}.

\section{Proving Equivalence of A Third Translation Strategy}
\label{sec:feedbackless}

The abstract algorithm for translating HBDs, as well as the two translation strategies presented in Section \ref{sec:impl}, use the $\fb$ operator when translating diagrams. 
As discussed in \cite{DBLP:conf/spin/DragomirPT16},
expressions that contain the $\fb$ operator are more complex to process
and simplify. For this reason, we wish to avoid using the $\fb$ operator as
much as possible.
Fortunately, in practice, diagrams such as those obtained from Simulink are {\em deterministic} and {\em algebraic loop free}.
As it turns out, such diagrams can be translated into algebraic expressions
that do not use the $\fb$ operator at all~\cite{DBLP:conf/spin/DragomirPT16}.
This can be done using the third translation strategy proposed in~\cite{DBLP:conf/spin/DragomirPT16}, called {\em feedbackless}.

While the two translation strategies presented in Section~\ref{sec:impl} can be modeled as refinements of the abstract algorithm, the feedbackless strategy is significantly more complex, and cannot be captured as such a refinement. We therefore treat it separately in this section. In particular,
we formalize the feedbackless strategy and we show that it is equivalent
to the abstract algorithm, namely, that for the same input, the results of the two algorithms are io-equivalent.

\subsection{Deterministic and Algebraic-Loop-Free Diagrams}

Before we introduce the feedbackless strategy, we need some additional definitions.

\begin{definition}
A diagram $S$ is {\em deterministic} if $$[x \leadsto x,x] \comp (S \ \| \ S) = S \comp [y \leadsto y , y].$$
An io-diagram $A$ is {\em deterministic} if $\dg(A)$ is deterministic.
\end{definition}
The definition of deterministic diagram corresponds to the following intuition. If we execute two 
copies of $S$ in parallel using the same input value $x$, we should obtain the same result as executing
one $S$ for the same input value $x$.

The deterministic property is closed under the serial, parallel, and switch operations of the HBD Algebra.
\begin{lemma} If $S,T\in\dgr$ are deterministic and $x,y$ are list of variables
such that $x$ are distinct and $\set(y) \subseteq \set(x)$, then
\begin{enumerate}
\item $[x\leadsto y]$ is deterministic\\[-1.5ex]
\item $S\comp T$ is deterministic\\[-1.5ex]
\item $S\parallel T$ is deterministic
\end{enumerate}
\end{lemma}

It is not obvious whether we can deduce from the axioms that the deterministic property is closed
under the feedback operation. However, since we do not use the feedback operation in this
algorithm, we don't need this property.

\begin{definition}
The {\em output input dependency relation} of an io-diagram $A$ is defined by
$$
 \oirel(A) = \set(\Out(A)) \times \set(\In(A))
$$
and the {\em output input dependency relation} of a list ${\cal A} = [A_1,\ldots, A_n]$
of io-diagrams is defined by
$$
  \oirel({\cal A}) = \oirel(A_1) \cup \ldots \cup \oirel(A_n)
$$
A list ${\cal A}$ of io-diagrams is {\em algebraic loop free}, denoted $\loopfree({\cal A})$, if
$$(\forall x: (x,x) \not\in (\oirel({\cal A}))^+ )$$
where $(\oirel({\cal A}))^+$ is the reflexive and transitive closure of relation $(\oirel({\cal A}))$.
\end{definition}
If we apply this directly to the list of io-diagrams from our example 
${\cal A} = [\add, \delay, \split]$ we obtain 
$$
    \oirel({\cal A}) = \{(x,u), (x, z), (y,x), (y, s), (s', x), (s', s), (z, y), (v,y)\}
$$
and we have that $(z,z) \in (\oirel({\cal A})) ^+$ because $(z,y), (y,x), (x,z) \in \oirel({\cal A})$,
therefore ${\cal A}$ is not algebraic loop free. However, the diagram from the example is accepted by Simulink, and it is considered algebraic loop free. In our treatment $\oirel({\cal A})$ contains
pairs that do not represent genuine output input dependencies. For example output $y$ of $\delay$
depends only on the input $s$, and it does not depend on $x$. Similarly, output $s'$ of $\delay$
depends only on $x$.

Before applying the feedbackless algorithm, we change the initial list of blocks into a new list such
that the output input dependencies are recorded more accurately, and all elements in the new list
have one single output. We split a basic block $A$ into a list of blocks $A_1,\ldots, A_n$ with 
single outputs such that $A \ioeq A_1 \Parallel \ldots \Parallel A_n$. Basically every block
with $n$ outputs is split into $n$ single output blocks.

We could do the splitting systematically by composing a block $A$ with all projections of the output. 
For example if $A=(x,(u_1,\ldots,u_n),S)$, then we can split $A$ into 
$A_i = (x,u_i, S \comp [u_1,\ldots,u_n \leadsto u_i])$.
Such splitting is always possible as shown in the following lemma:
\begin{lemma}
If $A$ is deterministic, then $A_1,\ldots, A_n$ is a splitting of $A$, i.e.
$$
A \ioeq A_1 \Parallel \ldots \Parallel A_n.
$$
\end{lemma}
However, this will still introduce unwanted 
output input dependencies. We solve this problem
by defining the splitting for every basic block, such that it accurately records the 
output input dependency. For example, we split the delay block into $\delay_1$ and $\delay_2$:
$$
\begin{array}{lll}
\delay_1 = (s,y,[s \leadsto s]) = (s,y,\id) \\[1ex]
\delay_2 = (x,s',[x \leadsto x]) = (x,s',\id)
\end{array}
$$
The $\split$ block is split into $\split_1$ and $\split_2$:
$$
\begin{array}{lll}
\split_1 = (y,z,[y \leadsto y]) = (y,z,\id)\\[1ex]
\split_2 = (y,v,[y \leadsto y]) = (y,v,\id)
\end{array}
$$
The blocks $\delay_1$, $\delay_2$, $\split_1$, and $\split_2$ are all the same, except the naming 
of the inputs and outputs.
The $\add$ block has one single output that depends on both inputs, so it remains unchanged.

After splitting, the list of single output blocks for our example becomes 
$${\cal B} = \big(\add, \delay_1, \delay_2, \split_1, \split_2\big)$$ 
and we have
$$
	\oirel({\cal B}) = \{(x,u), (x, z), (y, s), (s', x), (z, y), (v, y)\}.
$$
Now ${\cal B}$ is algebraic loop free.

\begin{definition}
A block diagram is {\em algebraic loop free} if, after splitting, the list of blocks is 
algebraic loop free.
\end{definition}

We assume that every splitting of a block $A$ into $B_1, \ldots, B_k$ is done such that 
$A \ioeq B_1 \Parallel \ldots \Parallel B_k$.
\begin{lemma}
\label{lem:splitting} 
If a list of blocks ${\cal A} = (A_1,\ldots,A_n)$ is split into 
${\cal B} = (B_1, \ldots, B_m)$, then we have
$$
A_1\Parallel \ldots \Parallel A_n  \ioeq B_1 \Parallel \ldots \Parallel B_m.
$$
\end{lemma}

\newcommand{\okfbless}{\mathsf{ok\_fbless}}

For the feedbackless algorithm, we assume that ${\cal A}$ is algebraic loop free, all
io-diagrams in $\cal A$ are single output and deterministic, and all outputs are distinct. 
We denote this by $\okfbless({\cal A})$.

\newcommand{\internal}{\mathsf{internal}}
\begin{definition}
For $\cal A$, such that $\okfbless({\cal A})$, a variable $u$ is {\em internal} in $\cal A$
if there exist $A$ and $B$ in $\cal A$ such that $\Out(A) = u$ and $u \in \set(\In(B))$.
We denote the set of internal variables of $\cal A$ by $\internal({\cal A})$.
\end{definition}

\newcommand{\CompA}{\rhd}
\begin{definition}
If $A$ and $B$ are single output io-diagrams, then their {\em internal serial composition} is
defined by
$$
A \CompA B = \mathsf{if\ } \set(\Out(A)) \subseteq \set(\In(B)) 
   \mathsf{\ then \ } A \Comp B \mathsf{\ else \ } B
$$
and 
$$
A \CompA (B_1,\ldots,B_n) = (A \CompA B_1,\ldots, A \CompA B_n)
$$
\end{definition}
We use this composition when all io-diagrams have a single output, and for an
io-diagram $A$, we connect $A$ in series with all io-diagrams from $B_1,\ldots,B_n$ that 
have $\Out(A)$ as an input.

The internal serial composition satisfies some properties that are used in proving
the correctness of the algorithm.
\begin{lemma} \label{lemma:internal:comute}
 If $\okfbless(A,B,C)$ then $( (A\CompA B) \CompA (A\CompA C) ) \ioeq ( (B \CompA A) \CompA (B\CompA C) )$
\end{lemma}
\begin{lemma} \label{lemma:fbless:step}
If $\okfbless({\cal A})$ and $A \in \set({\cal A})$ such that $\Out(A) \in \internal({\cal A})$ 
then
\begin{enumerate} 
  \item $\okfbless(A \CompA ({\cal A} \ominus A) )$ and \\[-1.5ex]
  \item $\internal(A \CompA ({\cal A} \ominus A)) = \internal({\cal A}) - \{\Out(A)\}.$
\end{enumerate}
\end{lemma}

\subsection{Functional Definition of the Feedbackless Strategy}

\newcommand{\fbless}{\mathsf{fbless}}
\begin{definition}
For a list $x$ of distinct internal variables of ${\cal A}$, we define by induction on
$x$ the function $\fbless(x,{\cal A})$ by
$$
\begin{array}{lll}
\fbless(\epsilon,{\cal A}) & = & {\cal A} \\[1ex]
\fbless(u\cdot x,{\cal A}) & = & \fbless(x, A \CompA ({\cal A} \ominus A))
\end{array}
$$
where $A$ is the unique io-diagrams from $\cal A$ with $\Out(A) = u$.
\end{definition}
Lemma~\ref{lemma:fbless:step} shows that the function $\fbless$ is well defined.

The function $\fbless$ is the functional equivalent of the feedbackless 
iterative algorithm that we introduce in Subsection~\ref{subsec:fblessalg}.
\begin{theorem}\label{theorem:fbless}
If ${\cal A} = (A_1,\ldots, A_n)$ is a list of io-diagrams satisfying 
$\okfbless({\cal A})$, $x$ 
is a distinct list of all internal variables of $\cal A$ ($\set(x) = \internal{\cal A}$), and
$(B_1,\ldots, B_k) = \fbless(x,{\cal A})$
then
$$
\FB(A_1 \Parallel \ldots \Parallel A_n) \ioeq (B_1 \Parallel \ldots \Parallel B_n).
$$
\end{theorem}
This theorem together with Lemma~\ref{lem:splitting} show that the result of the $\fbless$
function is io-equivalent to the results of the nondeterministic algorithm.
This theorem also shows that the result of $\fbless$ is independent of the choice 
of the order of the internal variables in $x$.

The proof of Theorem~\ref{theorem:fbless} is available in the RCRS formalization~\cite{RCRS_Toolset_figshare}, and it is
based on Lemmas~\ref{lemma:internal:comute} and~\ref{lemma:fbless:step} and other results.

\subsection{The Feedbackless Translation Algorithm}\label{subsec:fblessalg}
The recursive function $\fbless$ calculates the feedbackless translation, but it assumes
that the set of internal variables is given at the beginning in a specific order. 
We want an equivalent iterative version of this function, which at every step picks 
an arbitrary io-diagram $A$ with internal output, and performs one step:
$$
{\cal A} := A \CompA ({\cal A}\ominus A)
$$
The feedbackless algorithm is given in Alg.~\ref{alg:feedbackless}.
\begin{table}[H]
\renewcommand{\tablename}{\bf Alg.}
\centering
\normalsize
\begin{minipage}{0.9\textwidth}
{\sf
input: ${\cal A} = (A_{1}\ldots,A_{n})$ \ \  \ {\rm (list of io-diagrams 
satisfying $\okfbless({\cal A})$)}\\[1ex]
$\mathsf{while}\ \internal({\cal A}) \not= \emptyset$:\\[1ex]
\mbox{}~~~~~~$[\,{\cal A} := {\cal A'}\  |\ \exists\  
	 A \in \set({\cal A}) : \Out(A) \in \internal({\cal A})
  \,\land\,
   {\cal A}' = A \CompA ({\cal A}\ominus A)\,]$
    \\[1ex]
$A := B_1 \Parallel \ldots \Parallel B_k$ \rm \ \ \ (where 
${\cal A} = (B_1 , \ldots , B_k)$)
}\\
\end{minipage}
\caption{Feedbackless algorithm for translating HBDs.}
\label{alg:feedbackless}
\end{table}

The feedbackless algorithm is also nondeterministic, because it allows
choosing at every step one of the available io-diagrams with internal output.
As we will see in Subsection~\ref{subsec:efficiency}, this nondeterminism allows for 
different implementations regarding the complexity of the generated expressions.
\begin{theorem}
If ${\cal A} = (A_{1}\ldots,A_{n})$ is a list of io-diagrams satisfying 
$\okfbless({\cal A})$, then the feedbackless algorithm terminates for input ${\cal A}$,
and if $A$ is the output of the algorithm on $\cal A$, then
$$\FB(A_{1}\Parallel \ldots \Parallel A_{n}) \ioeq A.$$
\end{theorem}
\proof
Let $\mathsf{Feedbackless}$ be the predicate transformer of the feedbackless
algorithm. We prove that choosing nondeterministically
an order $x$ of the internal variables of ${\cal A}$, and 
calculating $\fbless(x,{\cal A})$ is refined by $\mathsf{Feedbackless}$. Formally we have:
$$
\begin{array}{lll}
&\{\okfbless({\cal A})\} \comp \\[1ex]
&\ \ \ \ [\,A := B_1 \Parallel \ldots \Parallel B_k \ |\ 
	\exists x: \set(x) = \internal({\cal A}) \land (B_1,\ldots,B_k) = \fbless(x,{\cal A})\,] \\[1ex]
	\sqsubseteq \\[1ex]
&	\mathsf{Feedbackless}
\end{array}
$$
To prove this refinement we need to use the assertion $\{\okfbless({\cal A})\}$. Intuitively,
this assertion restricts the refinement only for inputs $\cal A$ satisfying the property
$\okfbless({\cal A})$.

Because of this refinement, the feedbackless algorithm terminates. 

Using this refinement, Lemma~\ref{hoare:refinement} (connecting refinement to Hoare 
correctness triples),
and Theorem~\ref{theorem:fbless}, we obtain that
the output of $\mathsf{Feedbackless}$  satisfies the desired property:
$$\FB(A_{1}\Parallel \ldots \Parallel A_{n}) \ioeq A,$$
when the input satisfies $\okfbless({\cal A})$. Stated as a Hoare correctness triple, this property is:
$$
\big(\okfbless({\cal A}) \land {\cal A} = (A_1,\ldots, A_n)\big)\ \{\!|\ \mathsf{Feedbackless}\ |\!\}\ 
\big(\FB(A_{1}\Parallel \ldots \Parallel A_{n}) \ioeq A\big)
$$
The details are available in the RCRS formalization~\cite{RCRS_Toolset_figshare}.
\qed
\begin{theorem}
For a deterministic and algebraic loop free block diagram,
the feedbackless algorithm and the nondeterministic algorithm
are equivalent.
\end{theorem}
\proof
Assume ${\cal A} = (A_{1}\ldots,A_{n})$ is the initial set of blocks 
satisfying $\iodist({\cal A})$,
and $A$ is one possible output of the nondeterministic algorithm. We have
$\FB(A_{1}\Parallel \ldots \Parallel A_{n}) \ioeq A.$

Assume that ${\cal B} = (B_{1}\ldots,B_{m})$ is a splitting of $\cal A$
satisfying $\okfbless({\cal B})$ and $B$ is the output of the feedbackless algorithm
for $\cal B$. 
We have
$\FB(B_{1}\Parallel \ldots \Parallel B_{m}) \ioeq B.$

Because $\cal B$ is a splitting of $\cal A$ we also have
$
A_{1}\Parallel \ldots \Parallel A_{n} \ioeq B_{1}\Parallel \ldots \Parallel B_{m}.
$

Finally, using Lemma \ref{lem:ioeq-cong}, we obtain $A \ioeq B$.
\qed

If we apply the feedbackless algorithm to the example from Fig.~\ref{fig:sum-orig} we obtain:

$
\begin{array}{lll}
\\[1ex]
  & \big(\add, \delay_1, \delay_2, \split_1, \split_2\big) \\[1ex]
  \mapsto & \{\mbox{Variable $x$ is internal and $\Out(\add) = x$} \} \\[1ex]
  & \big(\delay_1, ((z,u), s', \dg(\add) \comp \dg(\delay_2)), \split_1, \split_2\big) \\[1ex]
  \mapsto & \{\mbox{Variable $y$ is internal and $\Out(\delay_1) = y$} \} \\[1ex]
  & \big(((z,u), s', \dg(\add) \comp \dg(\delay_2)), \\[1ex]
    & \qquad (s,z, \dg(\delay_1) \comp \dg(\split_1)), 
    (s,v,  \dg(\delay_1) \comp \dg(\split_2))\big) \\[1ex]
  \mapsto & \{\mbox{Variable $z$ is internal} \} \\[1ex]

  & \big(((s,u), s', ((\dg(\delay_1) \comp \dg(\split_1)) \ \| \ \id) \comp\dg(\add) \comp \dg(\delay_2)), \\[1ex] 
    & \qquad (s,v,  \dg(\delay_1) \comp \dg(\split_2))\big) \\[1ex]

  \mapsto & \{\mbox{There are no internal variables anymore} \} \\[1ex]
  
  & ((s,u), (s',v), [s,u \leadsto s,u,s] \comp  
      (((\dg(\delay_1) \comp \dg(\split_1)) \ \| \ \id) \comp\dg(\add) \comp \dg(\delay_2)) \\[1ex]
      & \qquad \ \| \ (\dg(\delay_1) \comp \dg(\split_2))
      
      )
\\[1ex]
\end{array}
$

$
\begin{array}{lll}
= & \{\mbox{Simplifications} \} \\[1ex]

  & ((s,u), (s',v), [s,u \leadsto s,u,s] \comp  
      (((\id \comp \id) \ \| \ \id) \comp\dg(\add) \comp \id) \ \| \ (\id \comp \id)
      )
\\[1ex]
= & \{\mbox{Simplifications} \} \\[1ex]

  & ((s,u), (s',v), [s,u \leadsto s,u,s] \comp  
      (\dg(\add) \ \| \ \id)
      )
      
\end{array}
$

\subsection{On the Nondeterminism of the Feedbackless Translation}\label{subsec:efficiency}
We have seen already that different choices in the nondeterministic abstract algorithm result in different algebraic expressions, e.g., with different numbers of composition operators. We show in this section that the same is true for the feedbackless translation algorithm.
In particular, consider a framework like the Refinement Calculus of Reactive Systems~\cite{DBLP:conf/spin/DragomirPT16}, where the intermediate results of the
algorithm are symbolically simplified at every translation step. Different choices
of the order of internal variables could result in different complexities of the simplification
work. We illustrate this with the example from Figure~\ref{fig:fbless:optimal}.
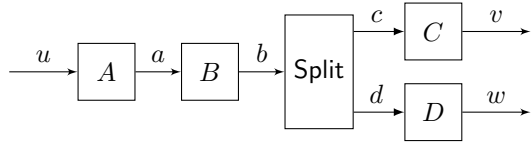
\begin{figure}
\begin{centering}
\begin{tikzpicture}
\node[draw, minimum height = 5ex, minimum width = 5ex](C){$C$};
\coordinate[below = 1ex of C](X);
\node[draw, minimum height = 5ex, minimum width = 5ex, below = 1ex of X](D){$D$};
\node[draw, minimum height = 10ex, minimum width = 5ex, left = 7ex of X](split){$\split$};
\node[draw, minimum height = 5ex, minimum width = 5ex, left = 4ex of split](B){$B$};
\node[draw, minimum height = 5ex, minimum width = 5ex, left = 4ex of B](A){$A$};
\draw[-latex'](A.east) -- ++(2ex,0)node[above]{$a$} --(B.west);
\draw[-latex'](B.east) -- ++(2ex,0)node[above]{$b$} --(split.west);
\draw[-latex](C.west -| split.east) -- ++(2ex,0)node[above]{$c$} -- (C.west);
\draw[-latex](D.west -| split.east) -- ++(2ex,0)node[above]{$d$} -- (D.west);
\draw[latex'-](A.west) -- ++(-3ex,0)node[above]{$u$} -- ++(-3ex,0);
\draw[-latex'](C.east) -- ++(3ex,0)node[above]{$v$} -- ++(3ex,0);
\draw[-latex'](D.east) -- ++(3ex,0)node[above]{$w$} -- ++(3ex,0);
\end{tikzpicture}
\par\end{centering}
\caption{Example for efficient implementation of feedbackless. \label{fig:fbless:optimal}}
\end{figure}

After splitting the list of blocks for this example is 
$$
{\cal A} = \big((u,a,A),\ (a,b, B),\ (b,c,\id), \ (b,d,\id), \ (c,v,C), \ (d,w, D)\big)
$$
and the set of internal variables is 
$$
\internal({\cal A}) = \{a,b,c,d\}.
$$

If we choose the order $(c,d,b,a)$, then after first two steps (including intermediate
simplifications) we obtain the list:
$$
\big((u,a,A),\ (a,b, B), \ (b,v,C), \ (b,w, D)\big)
$$
\newcommand{\simplify}{\mathsf{simplify}}
After another step for internal variable $b$ we obtain:
$$
\big((u,a,A),\ (a,v,\simplify(B \comp C)), \ (a,w, \simplify(B \comp D))\big)
$$
where the function $\simplify$ models the symbolic simplification. Finally, after applying the
step for the internal variable $a$ we obtain:
\begin{equation}\label{eq:first:order}
\big((u,v,\simplify(A\comp \simplify(B \comp C))), \ (u,w, \simplify(A\comp \simplify(B \comp D)))\big)
\end{equation}
In this order, we end up simplifying $A$ serially composed with $B$
twice. This is especially inefficient if $A$ and $B$ are complex. If we choose the
order $(c,d,a,b)$, then in the first three steps we obtain:
$$
\big((u,b,\simplify(A\comp B)), \ (b,v,C), \ (b,w, D)\big)
$$
At this point the term $A\comp B$ is simplified, and the simplified version is composed
with $C$ and $D$ to obtain:
\begin{equation}\label{eq:second:order}
\big((u,v,\simplify(\simplify(A\comp B)\comp C)), \ (u,w,\simplify(\simplify(A\comp B)\comp D))\big)
\end{equation}
If we compare relations (\ref{eq:first:order}) and (\ref{eq:second:order}) we see the same number of occurrences of $\simplify$, but in relation (\ref{eq:second:order}) there are
two occurrences of the common subterm $\simplify(A\comp B)$, and this is simplified only once.

As this example shows, different choices of the nondeterministic feedbackless translation strategy result in expressions of different quality, in particular
with respect to simplification.
It is beyond the scope of this paper to examine efficient deterministic implementations of the feedbackless translation.
Our goal here is to prove the correctness of this translation, by proving its
equivalence to the abstract algorithm.
It follows that every refinement/determinization of the feedbackless strategy will also be equivalent to the abstract algorithm, and therefore a correct implementation of the semantics.
Once we know that all possible refinements give equivalent results, we can concentrate in finding the most efficient strategy.
In general, we remark that this way of using the mechanisms of nondeterminism and refinement are standard in the area of correct by construction program development, and are often combined to separate the concerns of correctness and efficiency, as is done here.

\section{Implementation in Isabelle}

Our implementation in Isabelle uses locales \cite{Clemens:2014} for 
the axioms of the algebra. We use locale interpretations to show that these
axioms are consistent. In Isabelle locales are a powerful mechanism for developing consistent abstract theories (based on axioms). 
To represent the algorithm we use monotonic predicate transformers and we use Hoare total 
correctness rules to prove its correctness.

The formalization contains the locale for the axioms, a theory for constructive
functions, and one for proving that such functions are a model for the axioms. An 
important part of the formalization is the theory introducing the diagrams
with named inputs and outputs, and their operations and properties. 
The formalization also includes a theory for monotonic predicate transformers, refinement calculus, Hoare total correctness rules for programs, and a theory for the nondeterministic algorithm and its correctness.

In total the formalization contains 14797 lines of Isabelle code of which 
13587 lines of code for the actual problem, i.e., excluding the code for 
monotonic predicate transformers, refinement calculus, and Hoare rules.

\section{Conclusions and Future Work}
\label{sec:concl}

We introduced an abstract algebra for hierarchical block diagrams,
and an abstract algorithm for translating HBDs to terms of this algebra.
We proved that this algorithm is correct in the sense that no matter how its
nondeterministic choices are resolved, the results are semantically equivalent.
As an application, we closed a question left open 
in~\cite{DBLP:conf/spin/DragomirPT16} by proving that the Simulink
translation strategies presented there yield equivalent results.
Our HBD algebra is reminiscent of the algebra of flownomials~\cite{Stefanescu:2000:NA:518304} but our axiomatization is more general, in the sense that
our axioms are weaker.
This implies that all models of flownomials are also models of our algebra.
Here, we presented constructive functions as one possible model of our algebra.
Our work applies to hierarchical block diagrams in general, and the de facto
predominant tool for embedded system design, Simulink. 
Proving the HBD translator correct is a challenging problem, and as far
as we know our work is the only one to have achieved such a result.

We believe that our results are reusable in other contexts as well, in at least two ways. 
First, every other translation that can be shown to be a refinement/special case of our 
abstract translation algorithm, is automatically correct. For example, 
\cite{ReicherdtG2014,ZouZWFQ13} impose an order on blocks such that they use mostly 
serial composition and could be considered an instance of our abstract algorithm. 
Second, our algorithms translate diagrams into an abstract algebra. By choosing different 
models of this algebra we obtain translations into these alternative models.

As mentioned earlier, RCRS has been formalized in Isabelle~\cite{nipkow-paulson-wenzel-02}. The formalization is part of the RCRS toolset which is publicly available in a figshare repository~\cite{RCRS_Toolset_figshare}.
The theories relevant to this paper are under RCRS/Isabelle/TranslateHBD.
The RCRS toolset can be downloaded also from the RCRS web page: \url{http://rcrs.cs.aalto.fi/}.
The RCRS formalization represents a significant amount of work. The entire
formalization is close to 30000 lines of Isabelle code. The material for
this paper consists of 14797 lines of Isabelle code, 864 lemmas and 25 theorems,
and required an effort of 8 person-months excluding paper writing.

As future work we plan to investigate further HBD translation strategies, in addition to those studied above.
As mentioned earlier, this work is part of the broader RCRS project,
which includes a Translator of Simulink diagrams to RCRS theories
implemented on top of Isabelle~\cite{DBLP:conf/spin/DragomirPT16,DragomirPreoteasaTripakisFORTE2017,RCRSToolset_arxiv2017}.
Currently the Translator can only handle diagrams without algebraic loops,
i.e., without instantaneous circular dependencies. 
Extending the Translator and the corresponding determinacy proofs to
diagrams with algebraic loops is left for future work. This
is a non-trivial problem, because of subtleties in the definition of 
instantaneous feedback semantics, especially in the presence of
non-deterministic and non-input-receptive systems~\cite{preoteasa:tripakis:2016}.
For deterministic and input-receptive systems, however, the model of
constructive functions that we use in this paper should be sufficient.
Another future research goal is to unify
the proof of the third translation strategy with that of the
other two which are currently modeled as refinements of the abstract
translation algorithm.

This work covers hierarchical block diagrams in general and Simulink in
particular. Any type of diagram can be handled, however, we do assume
a {\em single-rate} (i.e., synchronous) semantics. Handling multi-rate
or event-triggered diagrams is left for future work. Handling
hierarchical state machine models such as Stateflow is also left for
future work.

Our work in this paper and in the RCRS project in general implicitly
provides, via the translation, a formal semantics for the subset of
Simulink described above. As already mentioned in~\S\ref{sec:rwork},
ultimately the semantics of Simulink is ``what the simulator does''.
Since the code of the simulator is proprietary, 
the only way to validate a formal semantics such as ours is by simulation.
Some preliminary work towards this goal is reported in~\cite{DBLP:conf/spin/DragomirPT16}, which also presents preliminary case studies, including
a real-world automotive control benchmark provided by Toyota~\cite{jin2014ARCH}.
A more thorough validation of the semantics and experimentation with
further case studies are future research topics.

As mentioned in~\S\ref{sec:rwork}, there are many existing translations
from Simulink to other formalisms. It is beyond the scope of this paper
to define and prove correctness of those translations, but this
could be another future work direction. In order to do
this, one would first need to formalize those translations. This in turn
requires detailed knowledge of the algorithms or even access to their
implementation, which is not always available.
Our work and source code are publicly available and we hope can serve
as a good starting point for others who may wish to provide formal
correctness proofs of diagram translations.

\paragraph{Acknowledgments.} We would like to thank Gheorghe \c Stef\u anescu 
for his help with the algebra of flownomials.
This work has been partially supported by the Academy of Finland and the U.S. National Science Foundation (awards \#1329759 and \#1139138).

\bibliographystyle{abbrv}
\bibliography{bibl}

\end{document}